\documentclass[sigconf]{acmart}

\usepackage{booktabs}
\usepackage{algorithm}
\usepackage{algorithmic}
\usepackage{graphicx}
\usepackage[table]{xcolor}
\usepackage{pifont}
\usepackage{subcaption}
\usepackage{threeparttable}
\usepackage{amsmath,amsfonts}
\usepackage{listings}
\usepackage{xspace}
\usepackage{makecell}
\usepackage{enumitem}
\usepackage{colortbl}
\usepackage{multirow}
\usepackage{wasysym} 
\usepackage{threeparttable}

\definecolor{randcol}{RGB}{214,183,165}
\definecolor{kindonlycol}{RGB}{190,143,112}
\definecolor{kindsemcol}{RGB}{123,125,103}
\definecolor{fullfpcol}{RGB}{165,162,132}
\definecolor{theorycol}{RGB}{205,203,188}

\lstdefinestyle{json}{
    basicstyle=\ttfamily\small,
    breaklines=true,
    showstringspaces=false,
    frame=single
}

\newcommand{\cmark}{\textcolor{teal}{\ding{51}\xspace}}
\newcommand{\xmark}{\textcolor{red}{\ding{55}\xspace}}
\newcommand{\pmark}{{\color{orange!80!black}$\boldsymbol{\circ}$}\xspace}
\newcommand{\sysname}{\textsc{PlanTwin}\xspace}
\newcommand{\piext}{\Pi_{\text{extract}}}
\newcommand{\pired}{\Pi_{\text{redact}}}
\newcommand{\pigen}{\Pi_{\text{generalize}}}
\newcommand{\pisch}{\Pi_{\text{schema}}}
\newcommand{\piout}{\Pi_{\text{out}}}

\renewcommand\footnotetextcopyrightpermission[1]{} 
\begin{document}

\title{\sysname: Privacy-Preserving Planning Abstractions for Cloud-Assisted LLM Agents}

\author{Guangsheng Yu$^{1}$, Qin Wang$^{1,2}$, Rui Lang$^{1}$, Shuai Su$^{1}$, Xu Wang$^{1}$}
\affiliation{
\textit{$^1$University of Technology Sydney} $|$ \textit{$^2$CSIRO Data61} 
\country{Australia}
}

\begin{abstract}

Cloud-hosted large language models (LLMs) have become the de facto planners in agentic systems, coordinating tools and guiding execution over local environments. In many deployments, however, the environment being planned over is private, containing source code, files, credentials, and metadata that cannot be exposed to the cloud. Existing solutions address adjacent concerns, such as execution isolation, access control, or confidential inference, but they do not control what cloud planners observe during planning: within the permitted scope, \textit{raw environment state is still exposed}.

We introduce \sysname, a privacy-preserving architecture for cloud-assisted planning without exposing raw local context. The key idea is to project the real environment into a \textit{planning-oriented digital twin}: a schema-constrained and de-identified abstract graph that preserves planning-relevant structure while removing reconstructable details. The cloud planner operates solely on this sanitized twin through a bounded capability interface, while a local gatekeeper enforces safety policies and cumulative disclosure budgets. We further formalize the privacy--utility trade-off as a capability granularity problem, define architectural privacy goals using $(k,\delta)$-anonymity and $\epsilon$-unlinkability, and mitigate compositional leakage through multi-turn disclosure control.

We implement \sysname\ as middleware between local agents and cloud planners and evaluate it on 60 agentic tasks across ten domains with four cloud planners. \sysname\ achieves full sensitive-item non-disclosure (SND = 1.0) while maintaining planning quality close to full-context systems: three of four planners achieve PQS $> 0.79$, and the full pipeline incurs less than 2.2\% utility loss.

\end{abstract}


\maketitle


\section{Introduction}
\label{sec:intro}

Large language models are increasingly used as planners in agentic systems, where a cloud-hosted model coordinates tools, decomposes tasks, and guides execution. Production systems such as Claude Code~\cite{anthropic2025claudecode}, OpenAI Codex~\cite{openai2025codex}, Cursor, and GitHub Copilot illustrate this design point: a powerful remote model reasons over code, files, and project context, while a local runtime executes the resulting plan. In many deployments, however, the environment being planned over is private. It contains source files, identifiers, logs, credentials, and sensitive metadata that cannot be exposed to the cloud. This creates a fundamental systems challenge: \emph{how can a local edge agent benefit from cloud-scale planning intelligence without revealing the real underlying state?}

Current systems address this challenge only partially. Claude Code's sandboxing architecture~\cite{anthropic2025sandbox} enforces filesystem and network isolation at the OS level, preventing access outside approved directories or connections to unapproved servers. OpenAI Codex~\cite{openai2025codex} executes tasks in a cloud sandbox with internet disabled, confining execution to a preloaded repository snapshot. The Model Context Protocol (MCP)~\cite{anthropic2024mcp} standardizes tool interfaces through schema-constrained descriptions and capability whitelisting. These mechanisms improve execution integrity, constrain side effects, and reduce blast radius. Some also coarsely restrict what the model can reach. But within the approved scope, \emph{the cloud planner still sees raw context}. File contents, code structure, and project metadata inside the permitted boundary are transmitted directly as prompt context. Existing mechanisms mostly control what the agent can \emph{do}, and only coarsely what it can \emph{reach}; they do not control the \emph{abstraction level} of what the cloud planner observes once access is granted.

This exposes a missing design point. Existing approaches largely fall into four categories: (1)~\emph{full-context upload}~\cite{anthropic2025claudecode,openai2025codex}, where raw or lightly filtered context is sent to the cloud and privacy relies on trust and contractual guarantees; (2)~\emph{Personally Identifiable Information (PII) redaction}~\cite{presidio2025,aws2025strands}, which removes recognized tokens before transmission but misses domain-specific secrets and leaves structural metadata that can still support re-identification; (3)~\emph{fully local inference}~\cite{privategpt2025,ollama2025}, which avoids cloud exposure but is constrained by local resources and cannot match the planning capability of frontier models; and (4)~\emph{trusted execution environments (TEEs)}~\cite{nvidia2024cc,akgul2024ollm,intel2024tdx}, which protect data during cloud inference but require specialized hardware and do not answer the semantic question of how much information the planner actually needs.

The central question is therefore not only how to secure execution, but how little information a cloud planner can observe while still producing useful plans. Our premise is that high-level orchestration rarely requires full descriptive fidelity of the underlying environment. For many workflows, useful planning depends on object types, relations, risk levels, action affordances, and dependency constraints, rather than raw file contents or precise identifiers. This suggests a different systems interface: the local environment should be exposed not as a compressed replica of reality, but as a \emph{planning-orientated control surface}.

We present \sysname, a \textit{privacy-preserving digital twin}\footnote{We use the term \emph{digital twin} deliberately. The core pattern is that a local system exports a machine-reasonable surrogate to a remote decision-maker.} architecture that realizes this idea. The local agent projects the real environment into a schema-constrained, de-identified, and planning-sufficient abstract representation through a four-stage projection pipeline ($\piext \to \pired \to \pigen \to \pisch$). The cloud planner operates only on this sanitized twin and emits declarative plans over a bounded set of local capabilities. All raw-data access, policy enforcement, execution, and output sanitization remain local. A Policy Enforcement Gatekeeper mediates every interaction through disclosure budgeting, safety validation, and optional human approval, ensuring that the cloud assists with planning but never directly accesses the real local world.

We complement recent capability-based control for LLM agents \cite{shi2025progent,zhu2025miniscope,ji2026seagent}. Those systems regulate \emph{what actions the agent may execute}; we regulate \emph{what information the planner may observe}. The two are orthogonal and naturally composable, combining execution-time access control with planning-time privacy control.
From this perspective, the central design question becomes the \emph{granularity of the capability interface}, which governs the privacy--utility trade-off. If capabilities are too coarse, most reasoning collapses back to the local side and the cloud provides limited planning benefit. If capabilities are too fine-grained, the planner regains utility only through repeated interaction, increasing cumulative disclosure. The desirable operating point lies near the \emph{orchestration boundary}, where the cloud determines workflow logic and the edge performs content-sensitive inspection and enforcement \cite{ma2025sok}.

\smallskip
\noindent \textbf{Contributions.} We make the following contributions:

\begin{itemize}[nosep]
\item We formulate \emph{cloud-assisted planning without raw-context exposure} as a system problem and position it against existing approaches including sandboxing, access control, PII redaction, and confidential computing (\S\ref{sec:related}).

\item We propose a four-stage local projection pipeline that converts the real environment into a sanitized, schema-constrained abstract graph. Unlike token-level masking, this design is \emph{structurally bounded}: the cloud planner never receives free-form raw context (\S\ref{sec:approach}).

\item We introduce a declarative capability-based planning protocol and formalize the privacy-utility trade-off as a \emph{granularity problem}, showing how capability design governs both planning effectiveness and disclosure cost (\S\ref{sec:approach:cap}).

\item We design a local gatekeeper that enforces cumulative disclosure budgets during planning, limiting how much structural information can be revealed to the cloud. We further formalize privacy goals for planning-time context control, ensuring that planner-visible abstractions do not uniquely identify underlying artifacts ($(k,\delta)$-anonymity) and cannot be reliably linked across sessions ($\epsilon$-unlinkability) (\S\ref{sec:approach:disclosure}).

\item We implement \sysname as middleware between local agents and cloud planners and show, on 60 tasks across ten domains with nineteen method configurations and four cloud planners, that strong cloud planning can be retained without exposing raw local context (\S\ref{sec:eval}).
\end{itemize}

\section{Background and Related Work}
\label{sec:related}

Cloud-assisted agent systems typically separate \emph{planning} from \emph{execution}: a powerful remote model reasons over user context and produces plans, while a local runtime executes tool calls or generated code. This architecture motivated several independent lines of work, including execution sandboxing, capability and access control, privacy-preserving context handling, and agent security. We summarize the qualitative design space in Table~\ref{tab:comparison}.

\begin{table}[t]
\centering
\caption{Privacy/security approaches for LLM agent systems.}
\label{tab:comparison}
\begin{threeparttable}
\resizebox{0.9\linewidth}{!}{%
\begin{tabular}{c|ccccc}
\cmidrule{2-6}
\multicolumn{1}{c}{\textbf{Approach}} & \ding{172} & \ding{173} &\ding{174} & \ding{175} & \ding{176} \\
\midrule
Claude Code Sandboxing~\cite{anthropic2025sandbox} & \cmark & \xmark & \xmark & \xmark & \cmark \\
OpenAI Codex Sandbox~\cite{openai2025codex} & \cmark & \xmark & \xmark & \xmark & \cmark \\
MCP Specification~\cite{anthropic2024mcp} & \cmark & \pmark & \cmark & \xmark & \cmark \\
E2B / Firecracker VMs~\cite{e2b2025sandbox} & \cmark & \xmark & \xmark & \xmark & \cmark \\
Progent~\cite{shi2025progent} & \cmark & \xmark & \xmark & \xmark & \cmark \\
MiniScope~\cite{zhu2025miniscope} & \cmark & \xmark & \xmark & \xmark & \cmark \\
SEAgent~\cite{ji2026seagent} & \cmark & \xmark & \xmark & \xmark & \cmark \\
PII Redaction (regex)~\cite{presidio2025} & \xmark & \pmark & \xmark & \xmark & \cmark \\
TEE (H100 CC, Intel TDX)~\cite{nvidia2023cc,intel2026tdx} & \cmark & \cmark & \xmark & \xmark & \xmark \\
Local-only~\cite{ollama2025,privategpt2025} & \cmark & \cmark & N/A & N/A & \cmark \\
\midrule
\rowcolor{kindsemcol!10}
\multicolumn{1}{c}{\textbf{\sysname (Ours)}} & \cmark & \cmark & \cmark & \cmark & \cmark \\
\cmidrule{2-6}
\end{tabular}%
}
\begin{tablenotes}
    \item  \ding{172} Controls agent actions,  \ding{173} Controls planner observations   
    \item  \ding{174}  Schema-bounded context  \ding{175}  Multi-turn disclosure  
    \item  \ding{176}  No specialized HW
\end{tablenotes}
\end{threeparttable}
\vspace{-0.1in}
\end{table}

\subsection{Agent Architecture and Execution Isolation}
\label{sec:related:sandbox}

Modern agentic coding tools have largely converged on a common architecture in which a cloud LLM reasons over code, files, and project context, while a local or containerized runtime executes the resulting plan. In this setting, the dominant systems concern has been \emph{execution isolation}.

Claude Code~\cite{anthropic2025sandbox} implements OS-level sandboxing using platform-specific primitives such as \texttt{sandbox-exec} on macOS and \texttt{seccomp} with \texttt{landlock} on Linux, providing filesystem and network isolation. OpenAI Codex~\cite{openai2025codex} executes tasks in a cloud sandbox over a preloaded repository snapshot, with internet access disabled. E2B~\cite{e2b2025sandbox} provides Firecracker microVM-based sandboxes with fast startup, and Agent Sandbox for Kubernetes~\cite{skypilot2025sandbox} relies on gVisor-based userland kernels with optional Kata Containers support.

These systems improve execution integrity and reduce blast radius once a plan is issued. However, their main protection boundary is \emph{what the agent can do}, not \emph{what the planner can observe}. Within the permitted scope, the cloud planner still commonly sees raw files, code semantics, and project metadata.

\subsection{Capability Interfaces and Access Control}
\label{sec:related:access}

A second line of work addresses \emph{authorization}: what tools and resources an agent may access. Progent~\cite{shi2025progent} introduces a domain-specific policy language for tool-level privilege control. MiniScope~\cite{zhu2025miniscope} enforces least-privilege authorization by reconstructing permission hierarchies with low runtime overhead. SEAgent~\cite{ji2026seagent} proposes mandatory access control through attribute-based policies and information-flow constraints. Li et al.~\cite{li2025visionaccesscontrolllmbased} further argue for richer policy models that can capture the semantic ambiguity of human--LLM interaction.

MCP~\cite{anthropic2024mcp} fits into this broader direction by standardizing schema-constrained tool descriptions and capability declarations. In this sense, it provides a partial \emph{capability interface}. However, MCP does not itself constrain the \emph{context} accompanying tool use: servers may still expose arbitrary raw resources, and planner-visible state remains largely unconstrained.

Recent studies show that schema-constrained tool interfaces alone do not prevent information leakage~\cite{zhao2025mindserversystematicstudy, zhang2025agentsecuritybenchasb}. Reported attacks include tool shadowing by malicious MCP servers~\cite{snyk2026agentscan}, Log-To-Leak~\cite{hu2026logtoleak}, which induces agents to invoke attacker-controlled logging tools, and MCP-SafetyBench~\cite{zong2026mcpsafetybench}, which reports substantial task compromise rates.  Prior evidence further suggests that models may still infer or propagate sensitive state through contextual interactions and tool outputs~\cite{wang2025your}. API- or schema-level constraints are insufficient. More broadly, representation-level transformations can preserve functionality while still changing observable signals and enabling downstream inference, so abstraction alone does not guarantee observation privacy~\cite{jiang2026why}. The distinction between controlling \emph{actions} and controlling \emph{observations} is explicit.

Our work targets the latter: it constrains what the planner sees before execution, which is orthogonal to and composable with execution-time access control, even when capabilities are exposed to the planner as high-level \textit{skill}-like operations~\cite{jiang2026sok}.

\subsection{Privacy-Preserving Context Mechanisms}
\label{sec:related:privacy}

A third line of work seeks to reduce privacy exposure in the context provided to LLMs~\cite{ma2025sok}.

\textit{PII redaction.}
Production systems such as Microsoft Presidio~\cite{presidio2025} and Strands Agents~\cite{aws2025strands} combine regex and NER to remove structured sensitive tokens. These approaches are useful for obvious PII, but they often miss domain-specific identifiers, secrets embedded in code or configuration files, and higher-order structural cues. They also typically operate on a per-request basis and do not track cumulative disclosure across turns.

\textit{Differential privacy.}
Differential privacy~\cite{dwork2006dp,abadi2016dpsgd} provides a principled way to bound information leakage through calibrated noise. However, applying DP directly to free-text prompts or rich project context remains difficult because of the high dimensionality and semantic sensitivity of natural language. In practice, DP is more naturally applied to training or aggregate analytics than to planning-time context exchange.

\textit{Embedding leakage.}
Prior work shows that learned representations can themselves leak sensitive information. Song and Raghunathan~\cite{song2020information} study attribute leakage from embeddings, while Morris et al.~\cite{morris2023text} and Li et al.~\cite{li2023sentence} demonstrate approximate text or sentence reconstruction. These results motivate our decision to keep retrieval and raw-context handling on the trusted local side rather than exposing vectorized surrogates to the cloud.

\textit{Confidential computing.}
TEE-based LLM inference~\cite{nvidia2023cc,ben2025distilled,intel2024tdx} protects data during cloud computation through hardware-enforced isolation. This addresses data-in-use protection against the cloud infrastructure, but not the semantic question central to our setting: how much information the planner should observe in the first place.

\textit{Local-only inference.}
Systems such as PrivateGPT~\cite{privategpt2025}, Ollama~\cite{ollama2025}, and related local-serving frameworks avoid cloud exposure entirely, but they also give up the planning capability of frontier cloud models, especially for long-horizon or multi-step reasoning tasks.

\smallskip
These approaches expose a missing design point. Existing work either keeps strong planning by revealing raw context, or preserves privacy by weakening or relocating the planner. Our goal is instead to retain strong remote planning while restricting the planner to a bounded, sanitized abstraction of the local environment.

\subsection{Agent Security Threats and Design Gap}
\label{sec:related:threats}

A growing body of work documents the security risks of LLM agents. A comprehensive SoK~\cite{peng2026sokpi} synthesizing 78 studies reports prompt injection success rates exceeding 85\%. Hung et al.~\cite{hung2026picode} analyze vulnerabilities specific to coding agents, including attacks mediated through tool ecosystems such as MCP. Real-world incidents include CVE-2025-53773~\cite{cve2025copilot}, a critical GitHub Copilot vulnerability, as well as prompt-injection attacks against cloud-based agent systems exploiting whitelisted integrations.

These attacks are amplified when the planner directly consumes raw local content. Once arbitrary files, logs, and metadata are admitted into the planning context, malicious instructions embedded in local artifacts can influence remote reasoning, and sensitive information can be exposed even if execution is later sandboxed. This observation motivates the design point of \sysname. Rather than relying solely on execution isolation, access control, or token-level redaction, we constrain the \emph{planner-visible observation surface} itself. The cloud planner reasons over a schema-constrained digital twin rather than raw local artifacts, reducing both privacy exposure and the attack surface associated with planning over untrusted content.

\section{System Model}
\label{sec:system}

We define the system boundary, threat model, and formal objective of \sysname. Table~\ref{tab:notation} summarizes notations used throughout.

\subsection{System Overview}
\label{sec:system:overview}

\begin{figure*}[t]
\centering
\includegraphics[width=0.9\textwidth]{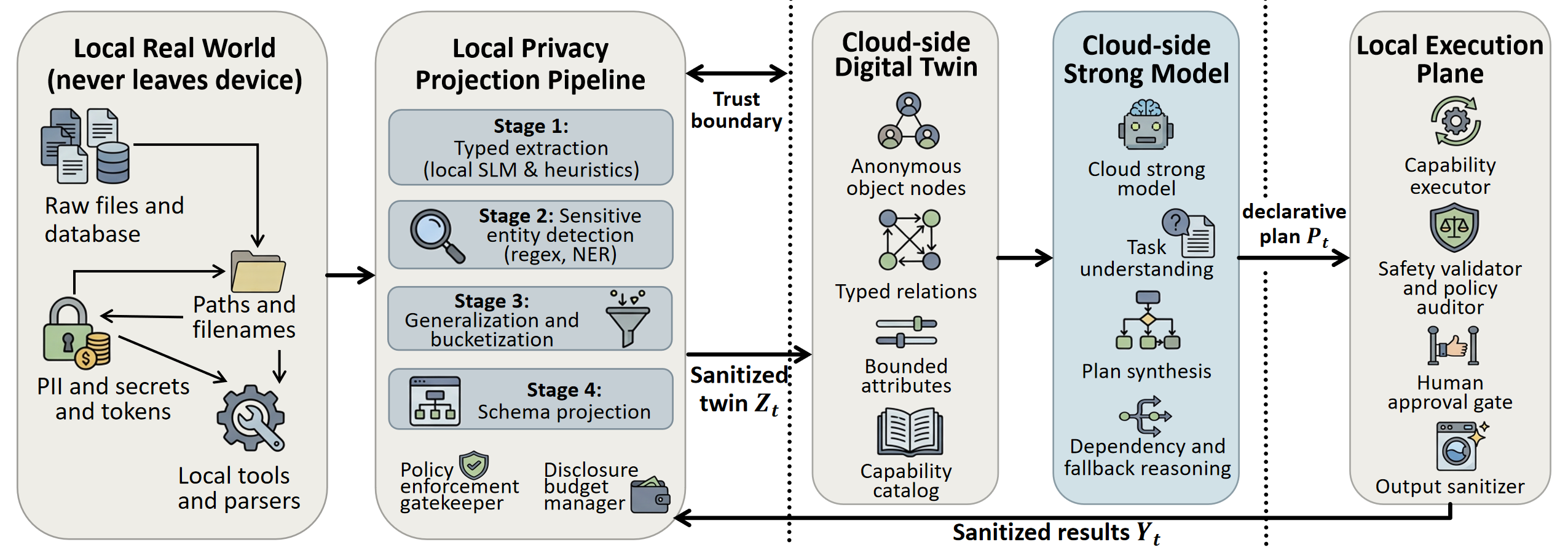}
\caption{End-to-end architecture of \sysname. The trusted local edge layer (left) transforms raw context $X_t$ through a four-stage privacy projection pipeline into a sanitized digital twin $Z_t$, which crosses the trust boundary to the cloud planner (center-right). The planner reasons over the typed abstract graph and capability catalog to produce a declarative plan $P_t$. The trusted local execution plane (right) validates each step through a gatekeeper, executes on raw data, and sanitizes outputs via $\piout$ before returning results $Y_t$. The cloud never observes raw files, paths, credentials, or code.}
\label{fig:system-overview}
\end{figure*}

\sysname follows a split architecture in which planning is delegated to a strong cloud model, while raw-state access and enforcement remain local. As shown in Fig.~\ref{fig:system-overview}, the system consists of two asymmetric components.

\begin{itemize}[nosep]
\item \textit{Trusted local layer.} The local side has direct access to the real environment and is responsible for privacy projection, policy enforcement, execution, and output sanitization. It may use deterministic scanners and optionally a small local language model to extract planning-relevant structure from raw artifacts.

\item \textit{Untrusted cloud planner.} The cloud side performs high-level reasoning and plan synthesis, but never accesses raw files, paths, credentials, or unsanitized execution outputs. Its view is restricted to the sanitized digital twin and the exposed capability interface.
\end{itemize}

The end-to-end information flow is
\begin{equation}
\label{eq:flow}
X_t \xrightarrow{\Pi} Z_t \xrightarrow{\Phi} P_t \xrightarrow{\Psi} Y_t,
\end{equation}
where $X_t$ denotes the real local environment, $\Pi$ the privacy projection pipeline, $Z_t$ the planner-visible digital twin, $\Phi$ cloud-side planning, $P_t$ the declarative plan emitted by the cloud planner, $\Psi$ local validation and execution, and $Y_t$ the sanitized output returned after local execution. This decomposition captures the central design choice of \sysname: the cloud plans over an abstract view, while all interaction with the real environment remains on the trusted local side. In this way, the cloud contributes planning capability without controlling what is revealed, executed, or returned.

\subsection{Threat Model}
\label{sec:system:threat}

We treat the cloud planner as untrusted with respect to privacy. Our primary model is \textit{honest-but-curious}: the planner follows the protocol for producing plans, but may attempt to infer sensitive information from everything it observes. We also consider a stronger \textit{semi-honest with auxiliary knowledge} adversary that may correlate multiple observations and combine them with external information.

Concretely, the adversary may observe all sanitized twin states $Z_{1:T}$ across turns, the sequence of requested capabilities and their abstract outputs, and timing or structural patterns arising during planning and execution. We also allow the adversary to combine these observations with auxiliary knowledge, such as public repositories, known project templates, or likely technology stacks.

The adversary does \emph{not} directly observe raw local files, metadata, secrets, or unsanitized outputs. The local layer, including the projection pipeline, gatekeeper, and output sanitizer, is part of the trusted computing base.

\smallskip
\noindent\textbf{Adversary goals.}
We focus on three privacy goals that capture the main ways a planner-visible abstraction may still leak information.

\textit{\underline{Goal 1}: Identity re-identification.}
Given transcript $Z_{1:T}$ and a target object $o^*$, an adversary outputs a guess $\hat{o}$ for the real-world identity $o^*$. We capture resistance to this attack via $(k,\delta)$-anonymity:
\begin{equation}
\label{eq:anonymity}
\Pr[\mathcal{A}(Z_{1:T}) = o^*] \leq \frac{1}{k} + \delta,
\end{equation}
where $k$ is the anonymity set size induced by schema projection. Appendix~\ref{app:proofs} formalizes the bound under assumptions on attribute distribution and disclosure budgeting.

\textit{\underline{Goal 2}: Cross-session linkage.}
Given transcripts from two sessions $Z^{(1)}_{1:T_1}$ and $Z^{(2)}_{1:T_2}$, the adversary determines whether the same real object appears in both. We capture resistance through $\epsilon$-unlinkability:
\begin{equation}
\label{eq:unlinkability}
\ln \frac{\Pr[\mathcal{A}(Z^{(1)}, Z^{(2)}) = 1 \mid \text{same}]}{\Pr[\mathcal{A}(Z^{(1)}, Z^{(2)}) = 1 \mid \text{diff}]} \leq \epsilon.
\end{equation}
where $\mathcal{A}$ is a linkage adversary, superscripts $(1)$ and $(2)$ index the two sessions, and $\text{same}$ (resp.\ $\text{diff}$) denotes the event that both transcripts contain the same (resp.\ distinct) real-world object.
Fresh random object identifiers, bounded vocabularies, and cumulative disclosure limits are intended to keep this quantity small. Appendix~\ref{app:proofs} connects unlinkability bound to disclosure budget.

\textit{\underline{Goal 3}: Structural inference.}
Even when content and identities are hidden, the adversary may infer sensitive properties such as project type, stack composition, or organizational structure from graph topology and capability patterns. This form of leakage is harder to characterize formally because structure itself is informative. We mitigate it through bounded edge vocabularies and controlled relational disclosure, but do not assume it can be eliminated entirely.

\smallskip
\noindent\textbf{Auxiliary knowledge.}
A realistic adversary may already know partial facts about the target environment. For example, knowing that the target is likely a Django project may make schema attributes such as \texttt{semantic\_class = authentication\_module} more identifying. Our goal is therefore not privacy against an omniscient adversary, but bounded exposure under realistic auxiliary knowledge and limited observation.

\subsection{Problem Formulation}
\label{sec:system:problem}

At time $t$, let $X_t$ denote the full local environment, including raw artifacts, metadata, logs, sensitive values, and interaction history. The local side does not expose $X_t$ directly. Instead, it applies the privacy projection pipeline $\Pi$ to construct a sanitized planner-visible state $Z_t = \Pi(X_t)$. Given $Z_t$ and capability catalog $\mathcal{C}$, the cloud planner produces a declarative plan $P_t = \Phi(Z_t, \mathcal{C})$, which the trusted local side then validates and executes under policy set $\Omega$ to produce the final outcome $Y_t = \Psi(P_t, X_t, \Omega)$.

The design objective is to retain as much planning utility as possible while keeping cumulative disclosure bounded:
\begin{equation}
\label{eq:objective}
\max \; \mathbb{U}(P_t, Y_t)
\quad \text{subject to} \quad
\mathcal{L}(Z_{1:t}, Y_{1:t}) \leq \epsilon.
\end{equation}
Here $\mathbb{U}$ denotes task utility and $\mathcal{L}$ measures cumulative exposure. In our design, disclosure is tracked per object through a budget $B(o)$, and $\mathcal{L}$ is conservatively taken as the maximum disclosure over all referenced objects, since over-exposing even a single object may already suffice for re-identification.

This formulation captures the core trade-off in \sysname: the planner should observe enough structure to generate useful plans, but not enough detail to reconstruct, identify, or reliably link the underlying local artifacts.

\begin{figure*}[t]
\centering
\includegraphics[width=\textwidth]{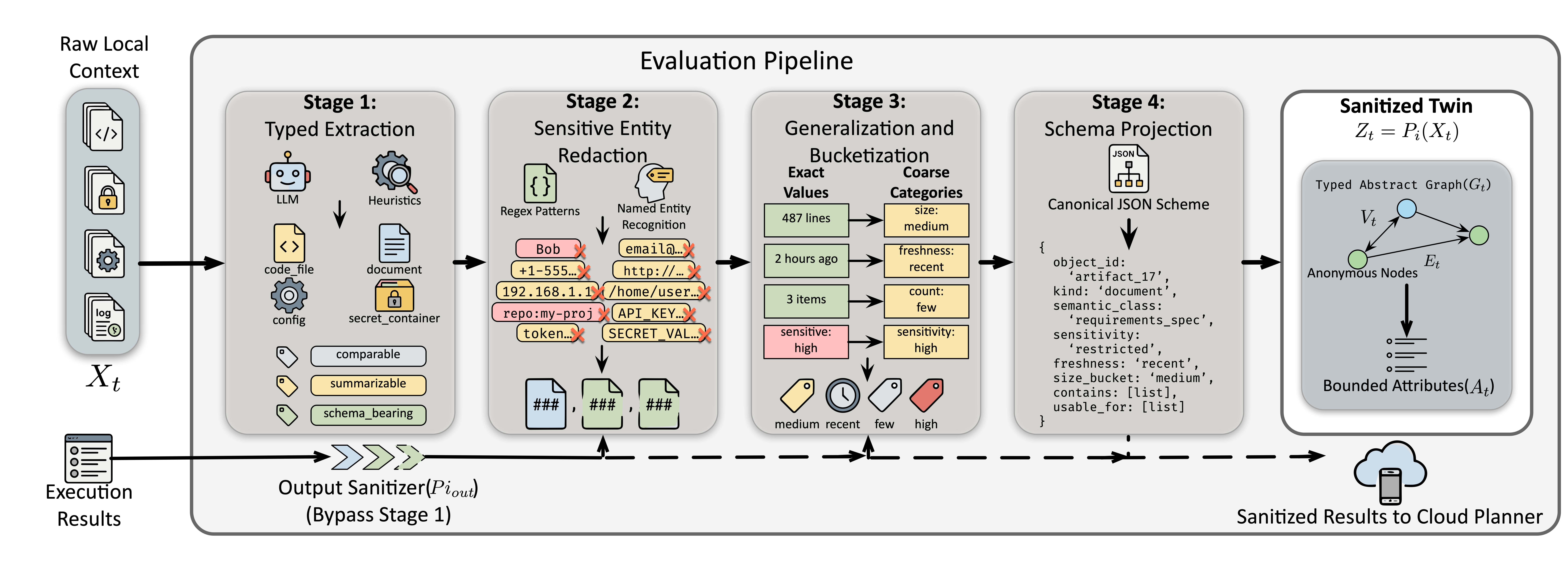}
\caption{The 4-stage privacy projection pipeline. Raw local context enters Stage~1 (typed extraction via local small language model (SLM) and heuristics), proceeds through Stage~2 (sensitive entity redaction via deterministic regex patterns with code-specific rules), Stage~3 (value generalization into bounded buckets), and Stage~4 (schema projection into fixed JSON). The output is a planning-sufficient digital twin $Z_t$ that contains no raw text, filenames, paths, or credentials. A companion output sanitizer $\piout$ applies Stages~2--4 to execution results before cloud exposure.}
\label{fig:projection-pipeline}
\end{figure*}

\section{Our Approach: \sysname}
\label{sec:approach}

\sysname is built around a simple idea: the cloud should help with \emph{planning}, but it should never need to observe the real local environment directly. To make this possible, the system combines three mechanisms. First, the local side converts raw context into a sanitized digital twin through a four-stage projection pipeline. Second, the cloud planner operates only through a bounded capability interface, so planning is expressed as declarative requests over abstract objects rather than direct access to local artifacts. Third, a local gatekeeper validates every step, enforces policy, and meters cumulative disclosure across turns.

These mechanisms separate \emph{reasoning} from \emph{exposure}. The cloud contributes planning intelligence over an abstract view, while the trusted local side retains control over raw-state access, execution, and what information is ultimately released.

\subsection{Privacy-Preserving Planning Abstraction}
\label{sec:approach:twin}

The first component of \sysname is a local projection pipeline that transforms the real environment into a planning-sufficient but reconstruction-resistant twin. Formally,
\begin{equation}
\label{eq:pipeline}
\Pi = \pisch \circ \pigen \circ \pired \circ \piext,
\end{equation}
that is, the local side first extracts structure, then redacts sensitive entities, then generalizes precise values, and finally projects the result into a fixed schema.

\smallskip
\noindent\textbf{Pipeline.} The role of this pipeline is not to summarize the environment as accurately as possible. Instead, it exposes only the structure needed for planning while suppressing content that would support reconstruction or re-identification.

\textit{\underline{Stage 1}: Typed extraction.}
The local extractor identifies object modality, coarse semantic role, and task affordances using a lightweight on-device small language model (SLM) with deterministic heuristics. Objects are mapped to coarse kinds such as \texttt{code\_file}, \texttt{log\_stream}, \texttt{structured\_record}, \texttt{config}, \texttt{secret\_container}, and \texttt{document}.
In parallel, the extractor assigns usability tags such as  \texttt{constraint\_source}, \texttt{requires\_local\_tool}, \texttt{schema\_bearing}, \\ \texttt{comparable} and \texttt{summarizable}. This stage determines what the object can be used for in planning, without yet deciding what may safely leave the trusted boundary.

\textit{\underline{Stage 2}: Sensitive entity detection and redaction.}
The second stage removes or masks obvious identifiers and secrets before any higher-level abstraction is exposed. This includes names, email addresses, phone numbers, URLs, IPs, hostnames, paths, repository identifiers, ticket IDs, keys, tokens, and other secrets. We use a deterministic detector combining regex rules for well-structured entities with keyword-based sensitivity classification for more contextual identifiers such as git metadata or internal hostnames.

\textit{\underline{Stage 3}: Generalization and bucketization.}
To reduce inversion risk, exact values are replaced by coarse categories. For example, byte sizes are mapped to \{small, medium, large\}, timestamps to \{today, recent, stale\}, counts to \{one, few, many\}, and sensitivity to \{low, restricted, high\}. This preserves planning-relevant distinctions while suppressing details that are unnecessarily identifying.

\underline{\textit{Stage 4}: Schema projection.}
Finally, the processed object is projected into a fixed JSON schema. A typical object takes the form:

\begin{lstlisting}[style=json]
{"object_id": "artifact_17",
 "kind": "document",
 "semantic_class": "requirements_spec",
 "sensitivity": "restricted",
 "freshness": "recent",
 "size_bucket": "medium",
 "contains": ["tables","dates","constraints"],
 "usable_for": ["compare","extract_constraints"],
 "confidence": 0.86}
\end{lstlisting}

This schema contains no filename, path, raw quote, or free-form text. The planner can reason over object roles and affordances, but not inspect the underlying content directly.

\textit{Output sanitization.}
The same principle also applies to execution results. We therefore define a companion output sanitizer $\piout$ that applies Stages~2--4 to local outputs before they are returned to the cloud. Stage~1 is omitted because result types are already known from the capability specification.

\textit{Security implication.}
Because raw local artifacts never enter the cloud-visible channel, the design also narrows the direct prompt-injection surface relative to raw-context systems. A remaining risk is that malicious local content may still influence Stage~1 extraction indirectly by causing the local SLM to misclassify objects. We mitigate this through deterministic cross-checks and by routing low-confidence extractions to review.

\smallskip
\noindent\textbf{Planning abstraction as a typed graph.}
\label{sec:approach:graph}
On the cloud side, the planner-visible abstraction is represented as a typed graph $G_t = (\mathcal{V}_t, \mathcal{E}_t, \mathcal{A}_t)$. Nodes in $\mathcal{V}_t$ correspond to anonymous artifacts (e.g., \texttt{artifact\_i}), edges in $\mathcal{E}_t$ encode coarse relations such as dependency, derivation, similarity, blocking, capability requirement, and conflict, and attributes in $\mathcal{A}_t$ capture bounded properties such as sensitivity, freshness, readiness, confidence, and usability tags. The output of $\Pi(X_t)$ is concretely this graph representation.

This graph view is important because the planner does not reason over isolated objects. It reasons over object roles, dependencies, and action affordances. The abstraction therefore preserves the relational structure needed for orchestration while withholding the raw artifacts from which that structure was derived.

\smallskip
\noindent\textbf{Running examples.}
\label{sec:approach:example}
Consider a coding-assistant task in which a developer asks: \emph{``Review the authentication module for security issues and compare it with the payment module.''}

In local environments, $X_t$ contains files such as \texttt{auth\_service.py} and \texttt{payment\_handler.py}, a \texttt{.env} file with live API keys, and git history containing developer identities. After projection, the twin $Z_t$ exposes only two code-file nodes with semantic classes such as ``authentication module'' and ``payment module,'' with a secret-container node and a small number of typed dependency edges. No filenames beyond these abstract roles, no developer identities, no API keys, and no git history cross the boundary.

The cloud planner therefore reasons over the abstract structure rather than the raw files. It may request a security audit of the two modules, compare the resulting findings, and generate an anonymized summary of the differences. These requests are then checked and executed locally on the real files.

\subsection{Capability-Based Planning}
\label{sec:approach:cap}

Once the planning abstraction has been constructed, the cloud planner no longer interacts with the local environment directly. Instead, it operates through a restricted capability interface $\mathcal{C} = \{c_1, \ldots, c_n\}$, where each capability is a safe typed operation such as constraint extraction, object comparison, issue classification, or security auditing. In this way, the planner reasons over abstract objects rather than raw files or unconstrained tool access.

The planner output is therefore not arbitrary code, but a structured action graph,
\begin{equation}
P_t = \{(s_i, c_i, I_i, O_i, \rho_i)\}_{i=1}^{K},
\end{equation}
where $s_i$ is a step identifier, $c_i$ a chosen capability, $I_i$ abstract inputs, $O_i$ the expected abstract output, and $\rho_i$ a policy annotation. This keeps planning declarative: the cloud specifies \emph{what operation should be performed}, while the trusted local side retains control over \emph{how} it is executed on raw data.

A central design choice is the granularity of this capability interface. If capabilities are too coarse, the local side must perform most of the reasoning itself and the cloud adds little value. If they are too fine-grained, the planner can recover utility only through repeated interaction, which increases cumulative disclosure. The privacy-utility trade-off is therefore governed not only by what capabilities exist, but by how much abstraction each capability exposes.

We capture this trade-off through $\text{grain}(\mathcal{C})$, the expected number of weighted schema fields revealed per capability call, using the disclosure model introduced in \S\ref{sec:approach:disclosure}. Let $\mathbb{U}(\mathcal{C})$ denote planning utility and $D(\mathcal{C},T)$ cumulative disclosure over $T$ turns. The resulting design objective is
\begin{equation}
\label{eq:granularity}
\mathcal{C}^* = \arg\max_{\mathcal{C}} \mathbb{U}(\mathcal{C})
\quad \text{s.t.} \quad
D(\mathcal{C},T) \leq \epsilon.
\end{equation}

In practice, the useful operating point lies near what we call the \emph{orchestration boundary}: capabilities should expose the workflow-level action to perform, but should not become a low-level channel for content inspection. Table~\ref{tab:granularity} illustrates this intuition.

\begin{table}[!]
\centering
\caption{Capability granularity trade-off.}
\label{tab:granularity}
\resizebox{\columnwidth}{!}{%
\begin{tabular}{c|l|c|c}
\toprule
\multicolumn{1}{c}{\textbf{Granularity}} &
\multicolumn{1}{c}{\textbf{Example}} &
\multicolumn{1}{c}{\textbf{Utility} $\uparrow$} &
\multicolumn{1}{c}{\textbf{Privacy Cost} $\downarrow$} \\
\midrule
Too coarse & \texttt{process\_project} & Low & Low \\
Orchestration & \texttt{security\_audit(art\_1)} & High & Medium \\
Too fine & \texttt{read\_line(art\_1, 42)} & High & High \\
\bottomrule
\end{tabular}%
}
\end{table}

\begin{table}[!]
\centering
\caption{Schema field disclosure costs.}
\label{tab:fieldcosts}
\resizebox{\columnwidth}{!}{%
\begin{tabular}{c|c|l}
\toprule
\multicolumn{1}{c}{\textbf{Field}} &
\multicolumn{1}{c}{$w_f$} &
\multicolumn{1}{c}{\textbf{Rationale}} \\
\midrule
\texttt{kind} & 0.5 & Low entropy, many objects share types \\
\texttt{semantic\_class} & 2.0 & Higher specificity \\
\texttt{contains} (per tag) & 1.0 & Content indicators \\
\texttt{usable\_for} (per tag) & 0.5 & Functional, less identifying \\
Typed edge (per edge) & 3.0 & Relational structure is highly identifying \\
\texttt{freshness} & 0.5 & Temporal bucket, coarse \\
\texttt{sensitivity} & 1.0 & Policy-relevant \\
\bottomrule
\end{tabular}
}
\end{table}

\subsection{Local Policy Enforcement}
\label{sec:approach:gatekeeper}

\begin{algorithm}[t]
\footnotesize
\caption{Gatekeeper Validation and Execution}
\label{alg:gatekeeper}
\begin{algorithmic}[1]
\REQUIRE Plan step $(s_i, c_i, I_i, O_i, \rho_i)$, context $X_t$, policy $\Omega$, budgets $\{B(o)\}$
\ENSURE Sanitized result $y_i$ or rejection signal
\IF{$c_i \notin \text{Allowlist}(\text{type}(I_i))$}
    \RETURN $\perp_{\text{reject}}$ \COMMENT{Capability not allowed}
\ENDIF
\FOR{each object $o \in I_i$}
    \STATE $\Delta \gets \text{EstimateCost}(c_i, o)$ \COMMENT{Weighted field cost}
    \IF{$B(o) + \Delta > \tau(o)$}
        \RETURN $\perp_{\text{reject}}$ \COMMENT{Budget exceeded}
    \ENDIF
\ENDFOR
\IF{$\rho_i$ requires \texttt{human\_approval}}
    \STATE $\text{approved} \gets \text{HumanGate}(s_i)$
    \IF{$\neg \text{approved}$}
        \RETURN $\perp_{\text{escalate}}$
    \ENDIF
\ENDIF
\STATE $r_i \gets \text{Execute}(c_i, I_i, X_t)$ \COMMENT{Local execution on raw data}
\STATE $y_i \gets \piout(r_i)$ \COMMENT{Sanitize output}
\FOR{each object $o \in I_i$}
    \STATE $B(o) \gets B(o) + \Delta_t(o)$ \COMMENT{Update budget}
\ENDFOR
\RETURN $y_i$
\end{algorithmic}
\end{algorithm}

\begin{algorithm}[t]
\footnotesize
\caption{\sysname: End-to-End Planning Cycle}
\label{alg:planning}
\begin{algorithmic}[1]
\REQUIRE Context $X_t$, projection $\Pi$, capabilities $\mathcal{C}$, policy $\Omega$, planner $\Phi$, thresholds $\{\tau(o)\}$
\ENSURE Sanitized result $Y_t$
\STATE Initialize $B(o) \gets 0$ for all tracked objects $o$
\FOR{each task request or planning cycle $t$}
    \STATE $X_t \gets \text{ObserveLocalState}()$
    \STATE $Z_t \gets \Pi(X_t)$ \COMMENT{4-stage projection: extract $\to$ redact $\to$ generalize $\to$ schema}
    \FOR{each object $o$ in $Z_t$}
        \IF{$B(o) + \Delta(o) > \tau(o)$}
            \STATE Suppress new fields or refuse disclosure
        \ENDIF
    \ENDFOR
    \STATE Send $Z_t, \mathcal{C}$ to cloud planner
    \STATE $P_t \gets \Phi(Z_t, \mathcal{C})$ \COMMENT{Receive declarative plan}
    \FOR{each step $(s_i, c_i, I_i, O_i, \rho_i)$ in $P_t$}
        \STATE $y_i \gets \text{GatekeeperValidateAndExecute}(s_i, X_t, \Omega)$ \COMMENT{Algorithm~\ref{alg:gatekeeper}}
        \IF{$y_i = \perp_{\text{reject}}$ \OR $y_i = \perp_{\text{escalate}}$}
            \STATE Handle rejection / escalation
        \ENDIF
    \ENDFOR
    \STATE $Y_t \gets \text{AggregateResults}(\{y_i\})$
    \STATE Log disclosure event to audit trail
\ENDFOR
\RETURN $Y_t$
\end{algorithmic}
\end{algorithm}

Cloud-generated plans are never executed directly. Every requested step is first mediated by a trusted local gatekeeper, which connects planning-time abstraction with execution-time enforcement. Its purpose is to ensure that each requested operation is valid for the referenced objects, consistent with local policy, and safe to execute on the real environment.

Operationally, the gatekeeper performs six checks in sequence: capability validation, policy validation, safety bounding, local execution, output sanitization, and escalation to a human when a step is ambiguous or high-risk. The full procedure is summarized in Algorithm~\ref{alg:gatekeeper}.

Formally, the local outcome is
\begin{equation}
\label{eq:gatekeeper}
Y_t =
\begin{cases}
\piout(\text{Exec}(P_t, X_t)) & \text{if policy-compliant},\\
\perp_{\text{reject}} & \text{if policy-violating},\\
\perp_{\text{escalate}} & \text{if ambiguous or high-risk}.
\end{cases}
\end{equation}

This boundary ensures that cloud planning remains advisory rather than authoritative. The cloud may propose an operation, but admissibility, execution over raw data, and any returned abstract result are all determined on the trusted local side. If a request is rejected or escalated, no new information is released to the cloud and $\Delta_t(o)=0$ for all involved objects 

This same boundary also limits plan-injection attempts. The cloud may request operations, but it cannot directly force raw-state disclosure or bypass local validation. Per-step controls such as capability whitelisting, output sanitization, and optional human approval bound what any single request may reveal, while cumulative controls such as disclosure budgeting and rate limits constrain what can be learned over repeated interaction.

\subsection{Multi-Turn Disclosure Control}
\label{sec:approach:disclosure}

Single-turn redaction is not sufficient in an interactive planning setting. Even when each individual response appears harmless, repeated queries may gradually reveal a distinctive fingerprint of an object. \sysname therefore treats disclosure as cumulative and tracks it across turns at the object level.

For each object $o$, cumulative disclosure is defined as
\begin{equation}
B(o) = \sum_{t=1}^{T} \Delta_t(o),
\end{equation}
where $\Delta_t(o)$ is the newly disclosed information at turn $t$. The gatekeeper checks every proposed release against an object-specific threshold $\tau(o)$ and blocks additional disclosure once the threshold would be exceeded.

\smallskip\noindent\textbf{Weighted field-cost model.}
We instantiate this mechanism with a simple weighted field-cost model in which each schema field $f$ is assigned a base disclosure weight $w_f$ reflecting its re-identification potential. Table~\ref{tab:fieldcosts} lists the default costs used in our prototype.

At turn $t$, the newly incurred disclosure is
\begin{equation}
\Delta_t(o) = \sum_{f \in \text{newly disclosed}} w_f.
\end{equation}
A new release is permitted only if
\begin{equation}
\label{eq:budget}
B(o) + \Delta_{T+1}(o) \leq \tau(o),
\end{equation}
where $\tau(o)$ may vary across object classes and sensitivity levels.

This mechanism is simple and operational. It is not presented as a full information-theoretic guarantee, but as a concrete local rule for limiting cumulative exposure in an interactive system.

\smallskip\noindent\textit{Advanced disclosure tracking.}
The weighted budget is only one possible instantiation. Stronger variants could quantify $\Delta_t(o)$ through entropy reduction, online $k$-anonymity monitoring, or differential privacy composition when individual disclosures satisfy suitable privacy guarantees. We mention these alternatives to clarify that the budgeting framework is general, even though the present implementation adopts a lightweight field-cost model.

\smallskip\noindent\textbf{Local retrieval and embedding exposure.}
All retrieval and similarity computation remain local. The cloud sees only abstract match outcomes or anonymized object identifiers, avoiding the embedding leakage and inversion risks demonstrated in prior work~\cite{song2020information,morris2023text,li2023sentence}.

\section{Evaluation}
\label{sec:eval}

We evaluate \sysname along several dimensions. We summarize our results in Table~\ref{tab:eval-roadmap}.

\begin{table}[!]
\centering
\caption{Evaluation roadmap.}
\label{tab:eval-roadmap}
\small
\renewcommand{\arraystretch}{1}
\vspace{-10pt}
\begin{tabular}{l|c|r}
\toprule
\multicolumn{1}{c}{\textbf{Evaluation dimensions}} & 
\multicolumn{1}{c}{\textbf{Ref.}} & 
\multicolumn{1}{c}{\textbf{Results}} \\
\midrule
Planning utility & \S\ref{sec:eval:utility} & Table~\ref{tab:main-results}, Fig.\ref{fig:privacy-utility} \\
Sensitive-item non-disclosure & \S\ref{sec:eval:utility} & Table~\ref{tab:main-results} \\
System overhead & \S\ref{sec:eval:overhead} & Table~\ref{tab:overhead} \\
Per-domain generalization & \S\ref{sec:eval:domain} & Table~\ref{tab:main-results} \\
Adversarial re-identification risk & \S\ref{sec:eval:reid} & Fig.\ref{fig:reid-strategies}, Fig.\ref{fig:reid-scaling} \\
Pipeline stage ablation & \S\ref{sec:eval:pipeline_ablation} & Fig.\ref{fig:ablation-pipeline} \\
Budget threshold sensitivity & \S\ref{sec:eval:budget_sensitivity} & Fig.\ref{fig:ablation-budget} \\
Multi-turn budget depletion & \S\ref{sec:eval:multiturn} & Fig.\ref{fig:ablation-multiturn-pqs}, Fig.\ref{fig:ablation-multiturn-nl} \\
Budget-footprint behavior & \S\ref{sec:eval:budget_analysis} & Table~\ref{tab:main-results} \\
\bottomrule
\end{tabular}
\end{table}

\subsection{Experimental Settings}
\label{sec:eval:setup}

\smallskip\noindent\textbf{Prototype.}
We implement \sysname as a Python middleware between a local agent runtime and a cloud LLM API. The system enforces a strict split between local projection and remote plan synthesis. By default, Stage~1 uses deterministic heuristic extraction for robustness and low latency, while a local Qwen3.5-family SLM~\cite{qwen2025qwen3} can optionally be used for richer semantic classification. Stage~2 uses regex-based entity detection with custom rules for code-specific secrets and identifiers, and Stages~3--4 are deterministic transformations. The disclosure budget manager maintains per-object state with per-field deduplication. Unless otherwise stated, all experiments use the default heuristic extraction pipeline.

\smallskip\noindent\textbf{Cloud planners.}
We evaluate four frontier cloud planners through OpenRouter, using the sanitized twin as the only planner-visible context: \textit{Kimi~K2.5}~\cite{kimik2_2025}, \textit{Gemini~3~Flash}, \textit{MiniMax~M2.5}~\cite{minimax2026m25}, and \textit{GLM~5}~\cite{glm2026glm5}. All planners are configured with high reasoning effort for consistency. All local processing, including projection and gatekeeper, runs on a single machine without GPU requirements.

\smallskip\noindent\textbf{Tasks.}
We construct a benchmark of 60 synthetic agentic tasks spanning ten domains: coding assistant, document review, multi-step debugging, DevOps, data pipelines, security incident response, ML operations, frontend development, database administration, and API integration. These tasks cover realistic agent workflows and embed sensitive artifacts such as PII, API keys, tokens, connection strings, internal hostnames, and infrastructure secrets.

\smallskip\noindent\textbf{Baselines.}
We compare \sysname against four baseline families. (i)~\textit{Raw Context} sends the full local environment to the cloud planner and serves as a utility upper bound without privacy protection. (ii)~\textit{PII Redaction} applies regex-based redaction before sending context to the cloud, but does not track cumulative disclosure across turns. (iii)~\textit{Local-Only} performs planning entirely with a local SLM without cloud assistance; we evaluate six models spanning three families and four scales: SmolLM2-360M-Instruct, Qwen3.5-0.8B, Qwen3.5-2B, Qwen3.5-4B, Qwen3.5-9B (with thinking mode), and Gemma-3-1B-it (instruction-tuned, no thinking mode). (iv)~\textit{Heuristic (Oracle)} uses manually curated keyword-to-capability rules without any cloud LLM, isolating the disclosure-budget mechanism from cloud planning quality. Combined with four cloud planners for \sysname, PII Redaction, and Raw Context, these baselines yield 19 method configurations in total.

\smallskip\noindent\textbf{\sysname configuration.}
Unless otherwise noted, \sysname uses heuristic Stage~1 extraction, capability-aware lazy disclosure, and a \emph{skill prompt} that teaches the cloud planner the twin's data format and capability semantics. We instantiate this configuration with each of the four cloud planners above (all with thinking mode enabled), yielding a total of 19 method configurations ($4~\text{planners} \times 3~\text{modes}$ + 6 local-only models + 1 heuristic) to systematically evaluate the privacy--utility--latency trade-off space.

\smallskip\noindent\textbf{Metrics.}
We report four metrics.  (i) \textit{Plan quality score (PQS)} $\uparrow$ is the task-averaged $F_1$ between approved and expected capabilities, capturing both over-generation and missing capabilities; we also report precision separately in the main results table. (ii) \textit{Sensitive non-disclosure (SND)} $\uparrow$ measures the fraction of sensitive items in the raw environment that do not appear in the planner-visible payload. (iii) \textit{Normalized leakage (NL)} $\downarrow$ measures the fraction of the disclosure budget consumed; for raw-context methods, we set NL$=1.0$ by convention. (iv) \textit{Latency overhead} measures the additional per-task cost introduced by projection and gatekeeper validation.


\subsection{Planning Utility and Privacy}\label{sec:eval:utility}

Table~\ref{tab:main-results} presents the main comparison across all nineteen configurations on the full 60-task benchmark.

\begin{table*}[!]
\centering
\caption{Main results across ten task domains (60-task evaluation per method)}
\label{tab:main-results}

\begin{threeparttable}
\resizebox{\linewidth}{!}{%
\begin{tabular}{c|c|cccccccccc|cccc}
\toprule
\multirow{2}{*}{\textbf{Planner}} & \multirow{2}{*}{\textbf{Method}} & \multicolumn{10}{c|}{\textbf{Per-Domain PQS}} & \multicolumn{4}{c}{\textbf{Combined} ($n$=60)} \\
\cmidrule(lr){3-12} \cmidrule(lr){13-16}
 & & \textbf{Code} & \textbf{Doc} & \textbf{Debug} & \textbf{DevOps} & \textbf{Pipe.} & \textbf{SecIR} & \textbf{MLOps} & \textbf{Front.} & \textbf{DB} & \textbf{API-I} & PQS $\uparrow$ & Prec $\uparrow$ & SND $\uparrow$ & NL $\downarrow$ \\
\midrule

\multicolumn{16}{c}{\textbf{Cloud LLM Planners}} \\
\midrule

Kimi~K2.5
  & Raw Context
  & 0.77 & 0.77 & 0.67 & 0.81 & 0.85 & 0.80 & 0.93 & 0.84 & 0.93 & 0.98
  & 0.843 & 0.831 & 0.000 & 1.000 \\
Kimi~K2.5
  & PII Redaction
  & 0.85 & 0.80 & 0.75 & 0.88 & 0.93 & 0.93 & 0.89 & 0.88 & 0.88 & 1.00
  & 0.887 & 0.886 & 0.835 & 0.165 \\
\rowcolor{kindsemcol!10}
Kimi~K2.5
  & \textbf{\sysname}
  & 0.79 & 0.80 & 0.68 & 0.72 & 1.00 & 0.74 & 0.79 & 0.67 & 0.90 & 0.91
  & 0.808 & 0.776 & 1.000 & 0.617 \\
\cmidrule(lr){1-2}\cmidrule(lr){3-12}\cmidrule(lr){13-16}

Gemini~3~Flash
  & Raw Context
  & 0.81 & 0.40 & 0.67 & 0.90 & 0.95 & 0.88 & 0.87 & 0.78 & 0.91 & 1.00
  & 0.832 & 0.833 & 0.000 & 1.000 \\
Gemini~3~Flash
  & PII Redaction
  & 0.89 & 0.40 & 0.71 & 0.95 & 0.89 & 0.91 & 1.00 & 0.90 & 0.83 & 0.95
  & 0.861 & 0.848 & 0.835 & 0.165 \\
\rowcolor{kindsemcol!10}
Gemini~3~Flash
  & \textbf{\sysname}
  & 0.71 & 0.68 & 0.57 & 0.72 & 0.93 & 0.48 & 0.68 & 0.70 & 0.66 & 0.77
  & 0.689 & 0.604 & 1.000 & 0.582 \\
\cmidrule(lr){1-2}\cmidrule(lr){3-12}\cmidrule(lr){13-16}

MiniMax~M2.5
  & Raw Context
  & 0.79 & 0.65 & 0.67 & 0.76 & 0.83 & 0.77 & 0.87 & 0.78 & 0.95 & 0.89
  & 0.802 & 0.761 & 0.000 & 1.000 \\
MiniMax~M2.5
  & PII Redaction
  & 0.65 & 0.75 & 0.73 & 0.77 & 0.78 & 0.79 & 0.87 & 0.75 & 0.91 & 0.83
  & 0.788 & 0.747 & 0.835 & 0.165 \\
\rowcolor{kindsemcol!10}
MiniMax~M2.5
  & \textbf{\sysname}
  & 0.75 & 0.69 & 0.67 & 0.83 & 1.00 & 0.70 & 0.86 & 0.73 & 0.72 & 0.81
  & 0.792 & 0.799 & 1.000 & 0.590 \\
\cmidrule(lr){1-2}\cmidrule(lr){3-12}\cmidrule(lr){13-16}

GLM~5
  & Raw Context
  & 0.77 & 0.40 & 0.60 & 0.89 & 0.90 & 0.87 & 1.00 & 0.90 & 0.82 & 0.98
  & 0.834 & 0.845 & 0.000 & 1.000 \\
GLM~5
  & PII Redaction
  & 0.77 & 0.40 & 0.67 & 0.94 & 0.95 & 0.95 & 1.00 & 0.93 & 0.78 & 0.94
  & 0.855 & 0.872 & 0.835 & 0.165 \\
\rowcolor{kindsemcol!10}
GLM~5
  & \textbf{\sysname}
  & 0.81 & 0.75 & 0.64 & 0.72 & 0.98 & 0.80 & 0.87 & 0.67 & 0.93 & 0.83
  & 0.810 & 0.765 & 1.000 & 0.615 \\
\cmidrule(lr){1-2}\cmidrule(lr){3-12}\cmidrule(lr){13-16}

\multicolumn{16}{c}{\textbf{Local Baselines (No Cloud)}} \\
\midrule

\rowcolor{gray!15} No cloud & Local (SmolLM2-360M)$^\ddagger$ & 0.10 & 0.50 & 0.50 & 0.12 & 0.38 & 0.50 & 0.17 & 0.33 & 0.17 & 0.17 & 0.283 & 0.567 & 1.000 & 0.000 \\
\rowcolor{gray!15} No cloud & Local (Qwen 3.5-0.8B)$^\ddagger$ & 0.10 & 0.50 & 0.50 & 0.12 & 0.38 & 0.50 & 0.17 & 0.33 & 0.17 & 0.17 & 0.283 & 0.567 & 1.000 & 0.000 \\
No cloud & Local (Gemma-3-1B)        & 0.55 & 0.86 & 0.50 & 0.54 & 0.60 & 0.57 & 0.55 & 0.55 & 0.55 & 0.44 & 0.567 & 0.559 & 1.000 & 0.000 \\
No cloud & Local (Qwen 3.5-2B)      & 0.10 & 0.50 & 0.50 & 0.12 & 0.38 & 0.50 & 0.17 & 0.33 & 0.17 & 0.17 & 0.283 & 0.567 & 1.000 & 0.000 \\
No cloud & Local (Qwen 3.5-4B)      & 0.27 & 0.40 & 0.50 & 0.46 & 0.58 & 0.50 & 0.52 & 0.70 & 0.54 & 0.30 & 0.484 & 0.656 & 1.000 & 0.000 \\
No cloud & Local (Qwen 3.5-9B)      & 0.10 & 0.50 & 0.86 & 0.12 & 0.64 & 0.93 & 0.73 & 0.50 & 0.45 & 0.60 & 0.531 & 0.642 & 1.000 & 0.000 \\
No cloud & Heuristic (Oracle)       & 0.53 & 0.50 & 0.50 & 0.78 & 0.90 & 0.70 & 0.59 & 0.54 & 0.43 & 0.84 & 0.653 & 0.782 & 1.000 & 0.558 \\

\bottomrule
\end{tabular}%
}
\begin{tablenotes}
\scriptsize
    \item[-] PQS = Plan Quality Score ($\uparrow$, $F_1$ of capability precision and recall), Prec = Capability Precision ($\uparrow$).
    \item[-] SND = Sensitive Non-Disclosure ($\uparrow$), NL = Normalized Leakage ($\downarrow$).
    \item[-] $^\ddagger$ Models $\leq 0.8$B fail to generate valid JSON plans; all PQS values reflect the single-capability fallback.
\end{tablenotes}
\end{threeparttable}
\end{table*}

\sysname achieves complete sensitive-item non-disclosure (SND =1.0) across all planners/domains: no API keys, passwords, PII, connection strings, SSH keys, or sensitive artifacts from the raw environment reach the cloud. This guarantee comes from architecture: the projection pipeline strips raw content before cloud exposure, and the planner only sees bounded categorical abstractions.

With Kimi~K2.5, \sysname reaches a combined PQS of 0.808 on the 60-task benchmark, compared with 0.843 for Raw Context. Two tasks produced degenerate outputs for all fourteen methods and were excluded from paired testing, leaving 58 common tasks. A paired Wilcoxon signed-rank test finds no significant difference between \sysname and Raw Context under Kimi ($T = 182.5$, $p = 0.079$, Bonferroni-corrected $p = 0.555$, $r = 0.230$), indicating that the planning abstraction preserves utility close to unrestricted context disclosure despite removing all raw content. This trend is consistent across planners: GLM~5 achieves PQS = 0.810 and MiniMax~M2.5 reaches 0.792, both statistically indistinguishable from Kimi. Gemini~3~Flash is weaker at PQS = 0.689, but also much faster, with 2.6$\times$ lower median latency (45.2\,s vs.\ 119.6\,s), highlighting a quality--latency trade-off.

\textit{PII redaction} achieves the highest plan quality (PQS=0.887 with Kimi) but only SND=0.835. Its failures are concentrated in coding tasks, where SND drops to 0.12 because regex-based redaction misses hardcoded connection strings and inline secret assignments. Across domains, SND ranges from 0.12 to 1.0, and overall PII redaction leaks 16.5\% of sensitive items, whereas \sysname leaks none.

Local-only planning achieves perfect non-disclosure (SND = 1.0) but plan quality is bounded by model capacity. Qwen3.5-2B produces PQS = 0.283, covering fewer than one of three expected capabilities per task; models at or below 2B (including SmolLM2-360M and Qwen3.5-0.8B) yield identical fallback plans. Scaling to Qwen3.5-4B improves PQS to 0.484 and Qwen3.5-9B reaches 0.531, but both exhibit high per-domain variance (9B: 0.10 on coding vs.\ 0.93 on security IR). Gemma-3-1B-it achieves the highest local-only PQS (0.567) with only 1B parameters, outperforming both Qwen3.5-4B and 9B while maintaining consistent cross-domain performance (0.44--0.86). Results indicate that instruction-tuning quality matters more than raw parameter count for structured plan generation. Nevertheless, the best local-only model (PQS=0.567) still falls 30\% below cloud-assisted \sysname (PQS=0.810 with GLM~5), showing that local SLMs alone cannot replace cloud-assisted planning.

The Heuristic (Oracle) baseline achieves PQS = 0.653 with SND = 1.0, but its per-domain scores vary widely (0.43--0.90) depending on how well the manually curated keyword rules match each domain's vocabulary; cloud-assisted \sysname consistently outperforms it on domains where keyword rules are weaker (coding: 0.79 vs.\ 0.53, debugging: 0.68 vs.\ 0.50, frontend: 0.67 vs.\ 0.54).

A Friedman test across all nineteen methods (repeated measures on 58 paired tasks) confirms globally significant differences in PQS distributions ($\chi^2 = 237.55$, $p < 0.001$, Kendall's $W = 0.315$). Pairwise Wilcoxon signed-rank tests with Bonferroni correction (18 comparisons vs.\ \sysname+Kimi) show that \sysname (Kimi) vs.\ Raw Context (Kimi) is \emph{not} significant ($p_{\text{Bonf}} = 1.0$, $r = 0.230$), and neither are Raw Context with Gemini ($p_{\text{Bonf}} = 1.0$), MiniMax ($p_{\text{Bonf}} = 1.0$), or GLM ($p_{\text{Bonf}} = 1.0$). Among \sysname variants, both MiniMax ($p_{\text{Bonf}} = 1.0$, $r = 0.034$) and GLM ($p_{\text{Bonf}} = 1.0$, $r = 0.049$) are statistically indistinguishable from Kimi, while Gemini is significantly lower ($p_{\text{Bonf}} = 0.004$, $r = 0.475$). PII Redaction (Kimi) achieves significantly higher PQS than \sysname (Kimi) ($p_{\text{Bonf}} = 0.013$, $r = 0.429$), the trade-off for its 16.5\% leakage. \sysname (Kimi) is significantly better than Heuristic (Oracle) ($p_{\text{Bonf}} < 0.001$, $r = 0.596$) and all six local-only models ($p_{\text{Bonf}} < 0.001$ for all, $r \geq 0.596$).

\begin{center}
   \colorbox{theorycol!25}{
    \begin{minipage}{0.9\linewidth}
\textbf{Privacy-utility.} \sysname is the only method that achieves complete sensitive-item protection (\textsc{snd}$=1.0$) while matching full disclosure on 3 of 4 cloud planners; PII redaction yields higher PQS but leaks 16.5\% of sensitive items.
   \end{minipage}
} 
\end{center}

Fig.\ref{fig:privacy-utility} visualizes the privacy-utility trade-off across all fourteen configurations, plotting SND (privacy) against PQS (utility) with per-task standard deviation error bars computed over all 60 tasks. The four \sysname configurations cluster in the upper-right quadrant (high SND, high PQS), achieving an operating point that no other method reaches. Raw Context methods cluster in the lower-right (high PQS, zero SND), while Local-Only occupies the upper-left corner (perfect SND, low PQS).

\begin{figure}[t]
\centering
\includegraphics[width=0.99\columnwidth]{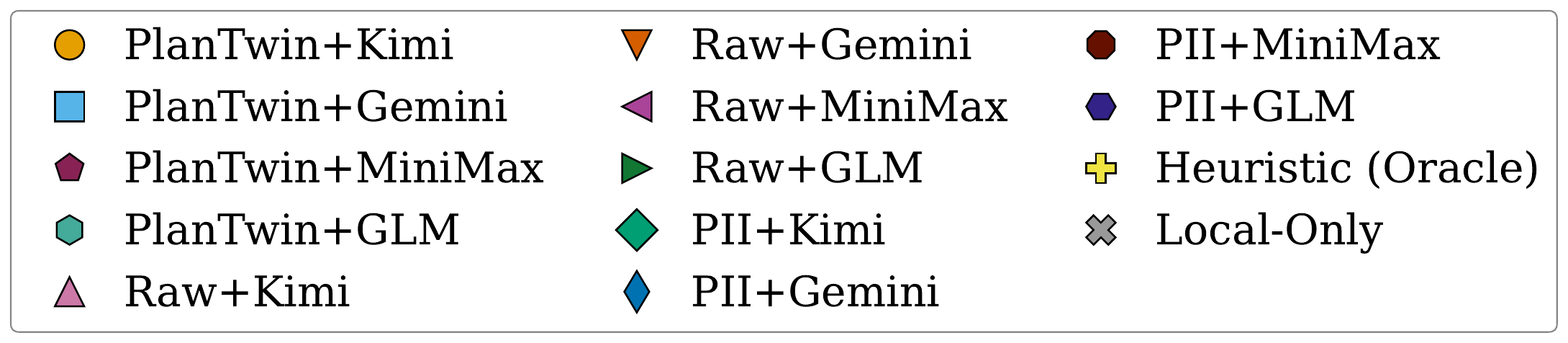}
\includegraphics[width=0.8\columnwidth]{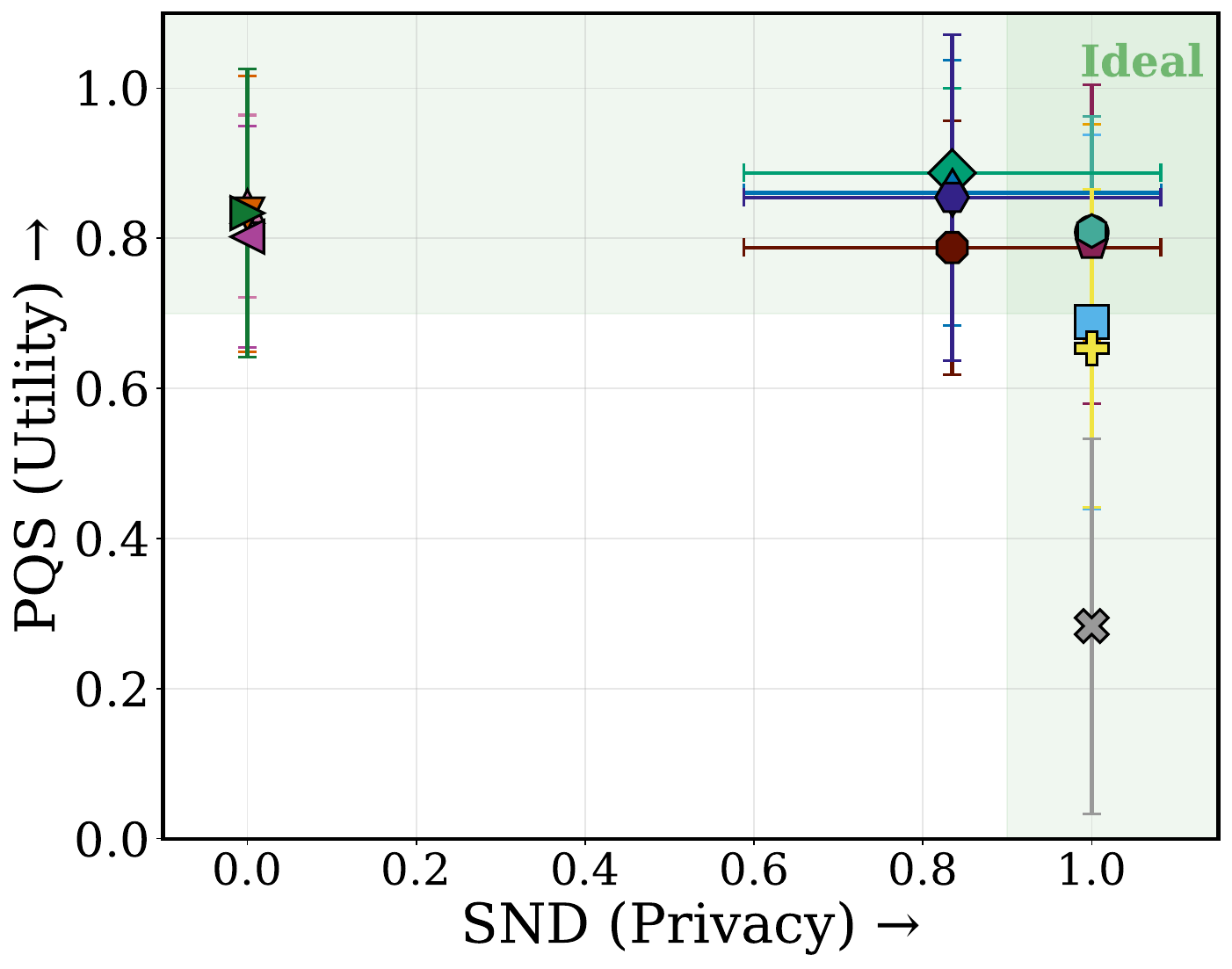}

\caption{Privacy-utility trade-off for 19 configurations over 60 tasks in 10 domains. Only \sysname reaches the ideal high-privacy/utility region, achieving SND=1.0 with PQS$>$0.7. Error bars show per-task standard deviation.}
\label{fig:privacy-utility}
\end{figure}

\subsection{Privacy Budget Analysis}\label{sec:eval:budget_analysis}

Average budget utilization across all 60 tasks reaches approximately 60\% for all four \sysname configurations (Kimi NL = 0.617, GLM NL = 0.615, MiniMax NL = 0.590, Gemini NL = 0.582), compared to 55.8\% for the Heuristic (Oracle) baseline. The tight NL range (0.582--0.617) across four diverse planners confirms that budget consumption is driven by the local gatekeeper's disclosure policy, not the cloud planner's behavior. Only the \texttt{kind} field is disclosed at registration (0.5 cost), and additional fields are disclosed on-demand per capability via CAPABILITY\_REQUIRED\_FIELDS. No object exceeds its threshold $\tau(o)$. PII redaction, by contrast, achieves low NL (0.165) only because it sends less structured metadata---but this comes at the cost of leaking raw content (SND = 0.835).

\begin{center}
   \colorbox{theorycol!25}{
    \begin{minipage}{0.9\linewidth}
\textbf{Budget effectiveness.} Capability-aware lazy disclosure keeps normalized leakage near 60\% across planners, preserving headroom for multi-turn interaction. Budget use is enforced locally by the gatekeeper.
   \end{minipage}
} 
\end{center}

\subsection{System Overhead}\label{sec:eval:overhead}

Table~\ref{tab:overhead} reports per-task latency broken down into \emph{local projection} (twin construction, including optional SLM-backed extraction) and \emph{cloud API} (planner round-trip). For baselines without schema projection (Raw Context, PII Redaction), the local component is negligible ($<$1\,ms); the entire latency is the cloud API call. For \sysname, the local projection phase includes Qwen~3.5-2B--backed extraction for Stage~1, adding 7--26\,s of SLM inference depending on task complexity. With heuristic-only extraction (no SLM), this drops to $<$1\,ms; the deterministic Stages~2--4 and gatekeeper contribute $<$1\,ms regardless. For local-only baselines, latency scales with model size: Qwen3.5-2B runs at median 6.7\,s, Gemma-3-1B at 14.3\,s, while Qwen3.5-4B (126.1\,s) and 9B (159.8\,s) incur high overhead due to thinking-mode token generation (max 8192 tokens).

\begin{table*}[!]
\centering
\caption{Per-task overhead breakdown (seconds)}
\label{tab:overhead}

\begin{threeparttable}
\resizebox{\linewidth}{!}{%
\begin{tabular}{c|c|cc|cccccccccc|c}
\toprule
\multirow{2}{*}{\textbf{Planner}} & \multirow{2}{*}{\textbf{Method}} & \multicolumn{2}{c|}{\textbf{Phase (mean)}} & \multicolumn{10}{c|}{\textbf{Per-Domain Total Latency}} & \multirow{2}{*}{\textbf{Med.}} \\
\cmidrule(lr){3-4} \cmidrule(lr){5-14}
& & \textbf{Proj.} & \textbf{API} & \textbf{Code} & \textbf{Doc} & \textbf{Debug} & \textbf{DevOps} & \textbf{Pipe.} & \textbf{SecIR} & \textbf{MLOps} & \textbf{Front.} & \textbf{DB} & \textbf{API-I} & \\
\midrule

\multicolumn{15}{c}{\textbf{Cloud LLM Planners}} \\
\midrule

Kimi~K2.5
  & Raw Context
  & --- & 53.0\,s & 67.9\,s & 101.5\,s & 52.7\,s & 88.0\,s & 105.4\,s & 111.2\,s & 95.4\,s & 48.4\,s & 248.2\,s & 137.9\,s & 66.8\,s \\
Kimi~K2.5
  & PII Redaction
  & $<$0.1\,s & 79.1\,s & 92.3\,s & 60.1\,s & 34.9\,s & 73.7\,s & 92.4\,s & 102.4\,s & 42.6\,s & 306.7\,s & 109.6\,s & 332.5\,s & 75.2\,s \\
\rowcolor{kindsemcol!10}
Kimi~K2.5
  & \textbf{\sysname}
  & 13.3\,s & 81.2\,s & 204.2\,s & 135.2\,s & 264.8\,s & 65.5\,s & 66.7\,s & 112.4\,s & 303.8\,s & 377.0\,s & 115.9\,s & 103.2\,s & 119.6\,s \\
\cmidrule(lr){1-2}\cmidrule(lr){3-4}\cmidrule(lr){5-14}\cmidrule(lr){15-15}

Gemini~3~Flash
  & Raw Context
  & --- & 6.5\,s & 17.8\,s & 11.5\,s & 12.9\,s & 13.6\,s & 12.2\,s & 19.9\,s & 19.3\,s & 15.5\,s & 17.1\,s & 13.4\,s & 14.4\,s \\
Gemini~3~Flash
  & PII Redaction
  & $<$0.1\,s & 8.5\,s & 23.0\,s & 12.8\,s & 19.1\,s & 14.9\,s & 18.8\,s & 28.8\,s & 16.3\,s & 17.6\,s & 23.7\,s & 27.4\,s & 19.3\,s \\
\rowcolor{kindsemcol!10}
Gemini~3~Flash
  & \textbf{\sysname}
  & 25.0\,s & 8.9\,s & 43.0\,s & 34.2\,s & 41.4\,s & 54.6\,s & 42.1\,s & 65.4\,s & 52.3\,s & 31.7\,s & 52.5\,s & 46.4\,s & 45.2\,s \\
\cmidrule(lr){1-2}\cmidrule(lr){3-4}\cmidrule(lr){5-14}\cmidrule(lr){15-15}

MiniMax~M2.5
  & Raw Context
  & --- & 28.6\,s & 14.1\,s & 19.9\,s & 30.5\,s & 16.5\,s & 95.9\,s & 15.8\,s & 17.0\,s & 9.7\,s & 20.0\,s & 24.7\,s & 16.6\,s \\
MiniMax~M2.5
  & PII Redaction
  & $<$0.1\,s & 44.8\,s & 29.4\,s & 135.0\,s & 22.6\,s & 28.8\,s & 30.4\,s & 16.0\,s & 16.7\,s & 115.4\,s & 13.9\,s & 55.3\,s & 15.3\,s \\
\rowcolor{kindsemcol!10}
MiniMax~M2.5
  & \textbf{\sysname}
  & 34.3\,s & 19.3\,s & 40.1\,s & 36.6\,s & 37.0\,s & 61.6\,s & 56.2\,s & 69.0\,s & 64.3\,s & 42.6\,s & 64.3\,s & 50.1\,s & 56.1\,s \\
\cmidrule(lr){1-2}\cmidrule(lr){3-4}\cmidrule(lr){5-14}\cmidrule(lr){15-15}

GLM~5
  & Raw Context
  & --- & 44.9\,s & 44.9\,s & 17.3\,s & 34.0\,s & 23.1\,s & 28.2\,s & 32.9\,s & 43.3\,s & 41.1\,s & 34.4\,s & 154.2\,s & 28.4\,s \\
GLM~5
  & PII Redaction
  & $<$0.1\,s & 66.6\,s & 44.0\,s & 36.7\,s & 44.5\,s & 48.6\,s & 53.4\,s & 61.3\,s & 43.3\,s & 72.2\,s & 206.5\,s & 49.9\,s & 48.0\,s \\
\rowcolor{kindsemcol!10}
GLM~5
  & \textbf{\sysname}
  & 33.4\,s & 32.1\,s & 46.0\,s & 34.6\,s & 38.0\,s & 70.2\,s & 60.8\,s & 102.3\,s & 67.6\,s & 49.5\,s & 91.6\,s & 77.4\,s & 63.3\,s \\
\cmidrule(lr){1-2}\cmidrule(lr){3-4}\cmidrule(lr){5-14}\cmidrule(lr){15-15}

\multicolumn{15}{c}{\textbf{Local Baselines (No Cloud)}} \\
\midrule

\rowcolor{gray!15} No cloud & Local (SmolLM2-360M)$^\ddagger$ & --- & 12.3\,s$^\dagger$ & 9.7\,s & 0.7\,s & 3.1\,s & 28.9\,s & 1.9\,s & 2.8\,s & 5.7\,s & 32.5\,s & 33.4\,s & 4.6\,s & 5.2\,s \\
\rowcolor{gray!15} No cloud & Local (Qwen 3.5-0.8B)$^\ddagger$ & --- & 3.6\,s$^\dagger$ & 6.1\,s & 2.6\,s & 2.8\,s & 4.4\,s & 2.8\,s & 2.8\,s & 5.6\,s & 2.9\,s & 2.9\,s & 2.9\,s & 2.9\,s \\
No cloud & Local (Gemma-3-1B)  & --- & 15.1\,s$^\dagger$ & 18.9\,s & 11.8\,s & 17.2\,s & 17.5\,s & 13.4\,s & 14.0\,s & 15.8\,s & 14.3\,s & 14.2\,s & 13.7\,s & 14.3\,s \\
No cloud & Local (Qwen 3.5-2B) & --- & 15.6\,s$^\dagger$ & 8.8\,s & 1.9\,s & 2.1\,s & 49.1\,s & 4.9\,s & 12.0\,s & 8.2\,s & 61.8\,s & 5.1\,s & 2.2\,s & 6.7\,s \\
No cloud & Local (Qwen 3.5-4B) & --- & 142.3\,s$^\dagger$ & 179.4\,s & 22.8\,s & 234.9\,s & 171.2\,s & 93.3\,s & 234.3\,s & 117.7\,s & 126.7\,s & 125.4\,s & 117.6\,s & 126.1\,s \\
No cloud & Local (Qwen 3.5-9B) & --- & 137.6\,s$^\dagger$ & 42.7\,s & 34.9\,s & 226.8\,s & 80.7\,s & 177.6\,s & 228.1\,s & 192.8\,s & 73.0\,s & 151.1\,s & 168.5\,s & 159.8\,s \\
No cloud & Heuristic (Oracle) & $<$0.1\,s & --- & $<$0.1\,s & $<$0.1\,s & $<$0.1\,s & $<$0.1\,s & $<$0.1\,s & $<$0.1\,s & $<$0.1\,s & $<$0.1\,s & $<$0.1\,s & $<$0.1\,s & $<$0.1\,s \\

\bottomrule
\end{tabular}%
}
\begin{tablenotes}
\scriptsize
\item[-] Proj. = mean local projection time (twin construction + optional SLM extraction), API = mean cloud planner API round-trip.
\item[-] Per-domain columns report mean total latency. All four cloud planners use thinking mode. Kimi~K2.5 exhibits high variance due to API load.
\item[-] $^\dagger$Local SLM inference (no cloud API); model size indicated in parentheses. ``---''\ = not applicable.
\item[-] $^\ddagger$Models $\leq$0.8B fail to generate valid JSON plans (fallback only).
\end{tablenotes}
\end{threeparttable}
\end{table*}

The phase reveals that local SLM projection is consistent across planners (13--34\,s), while cloud API latency varies substantially. For \sysname with Kimi~K2.5, the cloud API dominates (mean 81.2\,s, median 119.6\,s) with high variance due to API load. With MiniMax~M2.5 (thinking), the cloud API averages 19.3\,s, yielding a median total of 56.1\,s. GLM~5 (thinking) falls between Kimi and MiniMax at 32.1\,s mean API time (median total 63.3\,s). Gemini~3~Flash (thinking) is fastest at 8.9\,s mean API time (median total 45.2\,s).

The projection overhead is partially offset by prompt compression: on DevOps tasks, \sysname+Kimi (65.5\,s) is faster than Raw+Kimi (88.0\,s) because the compact twin reduces prompt tokens. With heuristic-only extraction (no SLM), the local phase drops to $<$1\,ms, reducing total latency to the cloud API time alone.

\begin{center}
   \colorbox{theorycol!25}{
    \begin{minipage}{0.9\linewidth}
\textbf{Practical overhead.} Total latency is dominated by SLM extraction (13--34\,s) and cloud API time; the deterministic projection pipeline adds $<1$\,ms. End-to-end latency ranges from 45\,s to 120\,s, and with heuristic extraction it is approximately the cloud API time.
   \end{minipage}
} 
\end{center}

Complexity analysis, token scaling, and deployment considerations are discussed in \S\ref{sec:discussion:complexity}.

\begin{figure}[b]
\centering
\includegraphics[width=0.8\columnwidth]{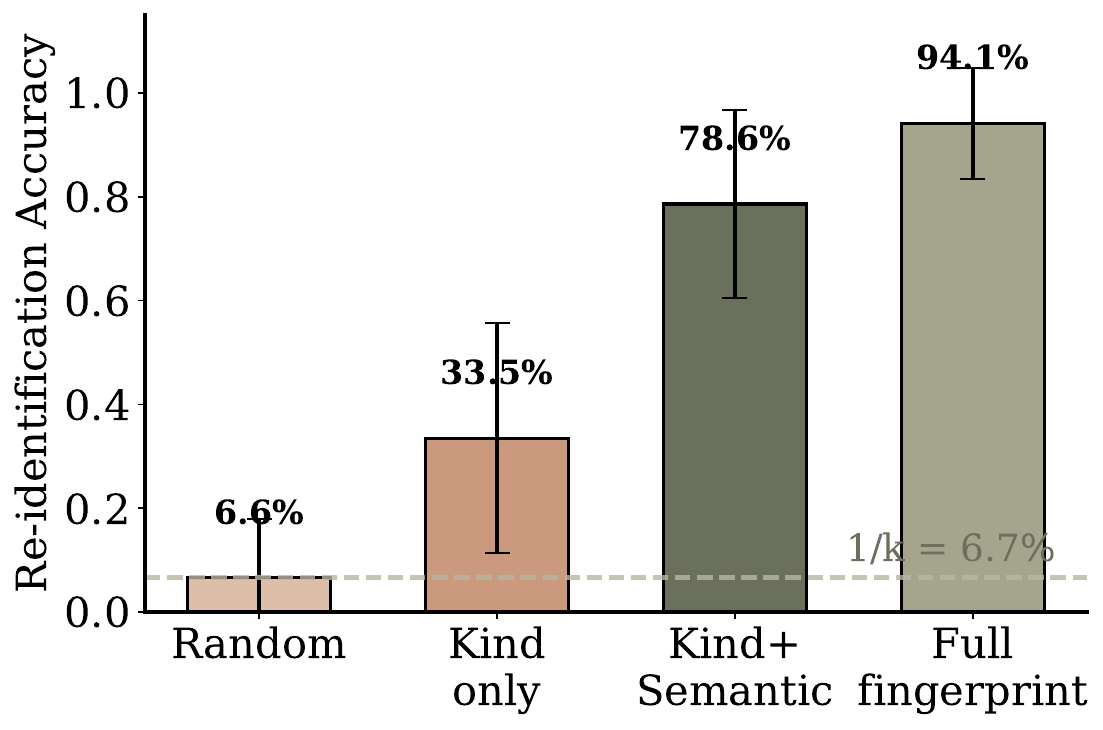}
\caption{Adversarial re-identification accuracy by strategy with $k=15$ candidate files. Error bars show standard deviation, and the dashed line marks the random baseline $1/k$. Full fingerprint matching reaches 94.1\% accuracy.}
\label{fig:reid-strategies}
\end{figure}

\subsection{Per-Domain Analysis}\label{sec:eval:domain}

Table~\ref{tab:main-results} presents per-domain PQS across all ten domains. With Kimi~K2.5, \sysname achieves PQS ranging from 0.67 (frontend) to 1.00 (data pipeline). Data pipeline tasks yield perfect PQS because their structured ETL workflows align naturally with the typed abstract graph representation. Database (0.90) and API integration (0.91) tasks also achieve high PQS, confirming that the projection pipeline generalizes well to structured integration workflows.

Cross-planner comparison reveals that Kimi, MiniMax, and GLM produce similar per-domain profiles: data pipeline (1.00/1.00/0.98), API integration (0.91/0.81/0.83), and database (0.90/0.72/0.93). The largest inter-planner variation occurs on security IR, where Kimi (0.74) and GLM (0.80) outperform MiniMax (0.70) and Gemini (0.48), suggesting that multi-step reasoning over the structured twin benefits from stronger planners. On DevOps tasks, all four planners are tightly clustered (0.72/0.72/0.83/0.72), while data pipeline tasks yield near-perfect PQS across the board, confirming that highly structured ETL workflows are planner-agnostic.

PII redaction's SND varies sharply by domain: coding tasks achieve only SND = 0.12 (regex misses 88\% of code-embedded secrets), while documents and API integration reach SND = 1.0. Database tasks show intermediate SND = 0.73, where regex partially captures connection strings but misses embedded PGP keys. \sysname maintains SND = 1.0 across all ten domains without exception, as the architectural guarantee is independent of sensitive-item distribution.

\begin{center}
   \colorbox{theorycol!25}{
    \begin{minipage}{0.9\linewidth}
\textbf{Domain generalization.} \sysname generalizes across ten domains and four cloud planners, sustaining \textsc{snd}$=1.0$ with PQS of 0.67-1.00. Consistent domain-level trends indicate planner-agnostic privacy-preserving projection.
   \end{minipage}
} 
\end{center}

\subsection{Adversarial Re-identification.}
\label{sec:eval:reid}

We empirically evaluate the re-identification risk of the twin's schema fields through a game-based experiment aligned with Goal~1 (\S\ref{sec:system:threat}). An adversary with auxiliary knowledge observes a target object's twin and attempts to match it to the correct file in a candidate pool of size $k$. We model four adversary strategies with increasing auxiliary knowledge: (1)~\emph{Random} guessing (baseline), (2)~\emph{Kind-only} matching using the file type field, (3)~\emph{Kind+Semantic} matching using both \texttt{kind} and \texttt{semantic\_class}, and (4)~\emph{Full fingerprint} matching using all disclosed twin fields (\texttt{kind}, \texttt{semantic\_class}, \texttt{contains}, \texttt{sensitivity}). Each strategy selects the candidate that maximizes field overlap with the target twin. We run 500 trials with pool size $k=15$ and 5 target objects per trial, drawing from a diverse template library spanning coding, configuration, documentation, and infrastructure files.

Fig.\ref{fig:reid-strategies} presents the results. Random guessing achieves 6.6\% accuracy, matching the theoretical $1/k = 6.7\%$. Knowing only the file type (\texttt{kind}) increases accuracy to 33.5\%, because the six kind categories already partition the pool into smaller anonymity sets. Adding the \texttt{semantic\_class} field raises accuracy to 78.6\%, since semantic labels such as \texttt{authentication\_module} or \texttt{payment\_service} are highly identifying. With the full fingerprint, accuracy reaches 94.1\% (MRR = 0.973), confirming that the twin's combined schema fields carry substantial re-identification potential when disclosed without budget control.

To evaluate how the anonymity set size affects re-identification risk, we repeat the experiment with pool sizes $k \in \{5, 10, 15, 20, 25, 30\}$ (300 trials each). Fig.\ref{fig:reid-scaling} shows the results. Random accuracy scales as $1/k$, decreasing from 19.3\% ($k=5$) to 4.3\% ($k=30$). Kind-only and Kind+Semantic strategies degrade with increasing $k$, dropping from 63.1\% and 94.1\% at $k=5$ to 19.7\% and 63.3\% at $k=30$, respectively. 
However, the full fingerprint adversary remains above 88.3\% even at $k=30$, indicating that the twin's combined fields are inherently identifying regardless of pool size. This result empirically motivates the disclosure budget mechanism (\S\ref{sec:approach:disclosure}): absent per-object budget limits, an adversary with auxiliary knowledge can re-identify files from the twin alone with high confidence.

\begin{center}
   \colorbox{theorycol!25}{
    \begin{minipage}{0.9\linewidth}
\textbf{Re-identification risk.} Without disclosure budget control, an adversary with full auxiliary knowledge achieves 94.1\% re-identification accuracy ($k=15$) and $>$88\% even at $k=30$. This empirically validates the need for per-object disclosure budgets: individual schema fields are not identifying, but their combination creates a unique fingerprint.
   \end{minipage}
} 
\end{center}

\subsection{Pipeline Stage Ablation}
\label{sec:eval:pipeline_ablation}

We isolate the contribution of each pipeline stage by running four variants: (1)~\emph{Full} (all four stages), (2)~\emph{No Redaction} (skip Stage~2, sensitive entity redaction), (3)~\emph{No Generalize} (skip Stage~3, value generalization), and (4)~\emph{Extract+Schema} (Stages~1--2 only, extract metadata and build typed graph but skip generalization and schema projection). All variants use Kimi~K2.5 with plan-aware prompting on the same 58-task benchmark.

Fig.\ref{fig:ablation-pipeline} presents the results. Across all four variants, SND remains 1.0, confirming that the typed abstract graph representation provides architectural data isolation regardless of whether downstream redaction or generalization stages are active. PQS ranges from 0.773 (Full) to 0.790 (No Redaction), a spread of only 0.017. Skipping redaction yields a marginal PQS improvement (+2.2\%) because the planner receives slightly richer metadata fields, while skipping generalization has negligible effect (+0.2\%). The Extract+Schema variant (0.788) performs comparably to Full, suggesting that the core utility derives from the structured twin representation itself rather than the privacy-hardening stages.

\begin{figure}[!]
\centering
\includegraphics[width=0.8\columnwidth]{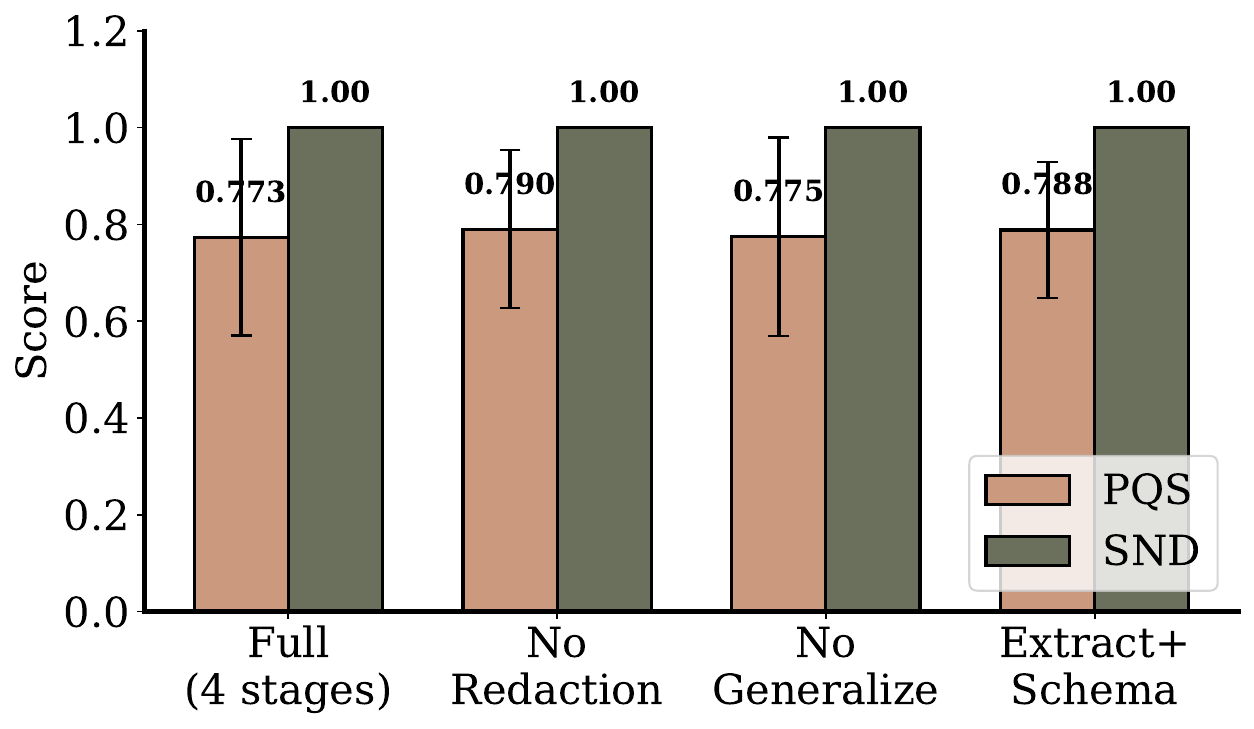}
\caption{Pipeline stage ablation across four variants. All maintain SND = 1.0, while PQS stays within 0.773--0.790, indicating minimal utility loss from privacy hardening. Error bars show per-task standard deviation.}
\label{fig:ablation-pipeline}
\end{figure}

Per-domain analysis reveals that data pipeline tasks achieve PQS = 1.00 across all four variants, while coding tasks show the largest inter-variant spread (0.534--0.753), reflecting their higher sensitivity to metadata granularity. Normalized leakage (NL) is lowest for No Redaction (0.553) and Extract+Schema (0.566), as fewer pipeline stages consume less disclosure budget per object.

\begin{center}
   \colorbox{theorycol!25}{
    \begin{minipage}{0.9\linewidth}
\textbf{Pipeline ablation.} The privacy-hardening stages (redaction, generalization) impose $<$2.2\% PQS cost while the typed abstract graph representation alone guarantees SND = 1.0. The full 4-stage pipeline offers the best privacy--utility balance: minimal utility loss with maximum defense-in-depth.
   \end{minipage}
} 
\end{center}

\subsection{Budget Threshold Sensitivity}
\label{sec:eval:budget_sensitivity}

We vary the disclosure budget threshold by applying scale factors $\alpha \in \{0.25, 0.5, 0.75, 1.0, 1.5\}$ to the default per-object budget $\tau(o)$, yielding effective thresholds $\alpha \cdot \tau(o)$. Lower $\alpha$ restricts disclosure more aggressively; higher $\alpha$ permits more fields to be revealed.

Fig.\ref{fig:ablation-budget} shows the privacy--utility trade-off. At $\alpha = 0.25$ (quarter budget), PQS drops to 0.513 as the gatekeeper aggressively withholds twin fields, starving the planner of context. Doubling the budget to $\alpha = 0.5$ recovers PQS to 0.639. At $\alpha = 0.75$, PQS reaches 0.801, matching the default ($\alpha = 1.0$, PQS = 0.801) and indicating that the default threshold already provides near-saturated utility. Increasing to $\alpha = 1.5$ yields PQS = 0.783, a slight decrease likely due to stochastic variation rather than a meaningful trend; SND remains 1.0 throughout.

Normalized leakage exhibits an inverted-U pattern: NL peaks at $\alpha = 0.75$ (0.777) then decreases at $\alpha = 1.0$ (0.626) and $\alpha = 1.5$ (0.466). This occurs because larger budgets leave more unused headroom per object, reducing the fraction of budget consumed. The practical implication is that $\alpha \in [0.75, 1.0]$ represents a sweet spot: PQS is saturated while NL remains moderate.

\begin{figure}[!]
\centering
\includegraphics[width=0.75\columnwidth]{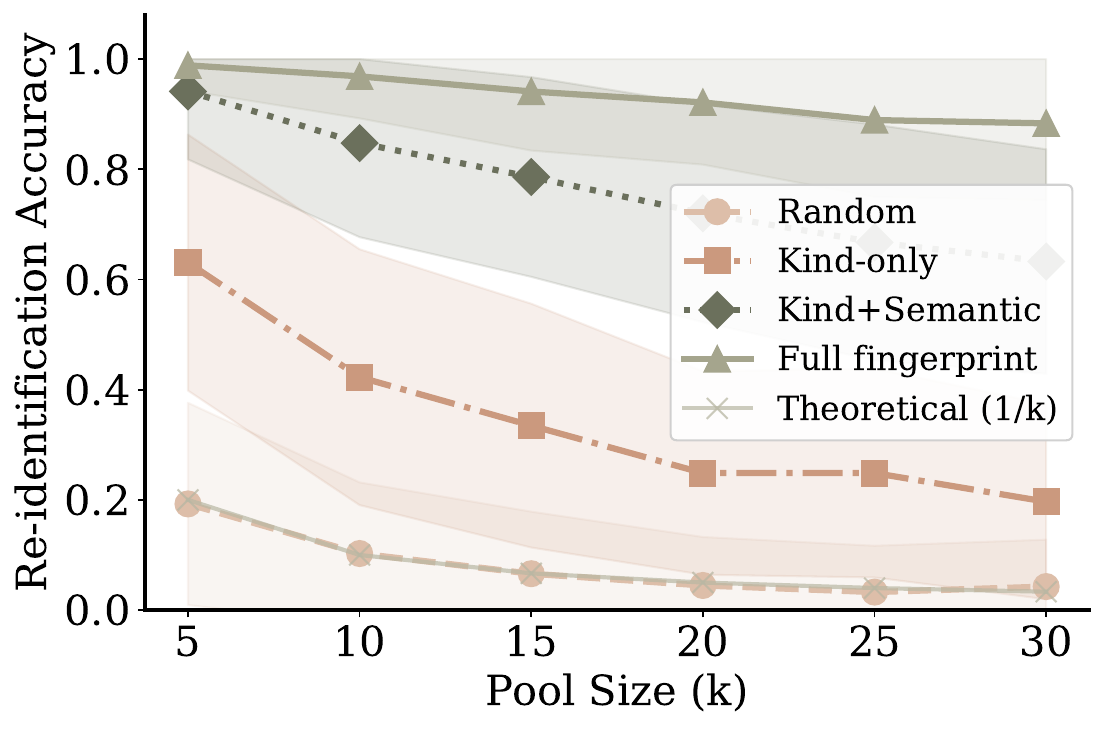}
\caption{Re-identification accuracy versus anonymity set size $k$. Full fingerprint matching remains above 88\% even at $k=30$, while weaker strategies degrade with larger pools. Shaded regions denote $\pm 1$ std.\ dev.}
\label{fig:reid-scaling}
\vspace{-0.1cm}
\end{figure}

\begin{figure}[!]
\centering
\includegraphics[width=0.8\columnwidth]{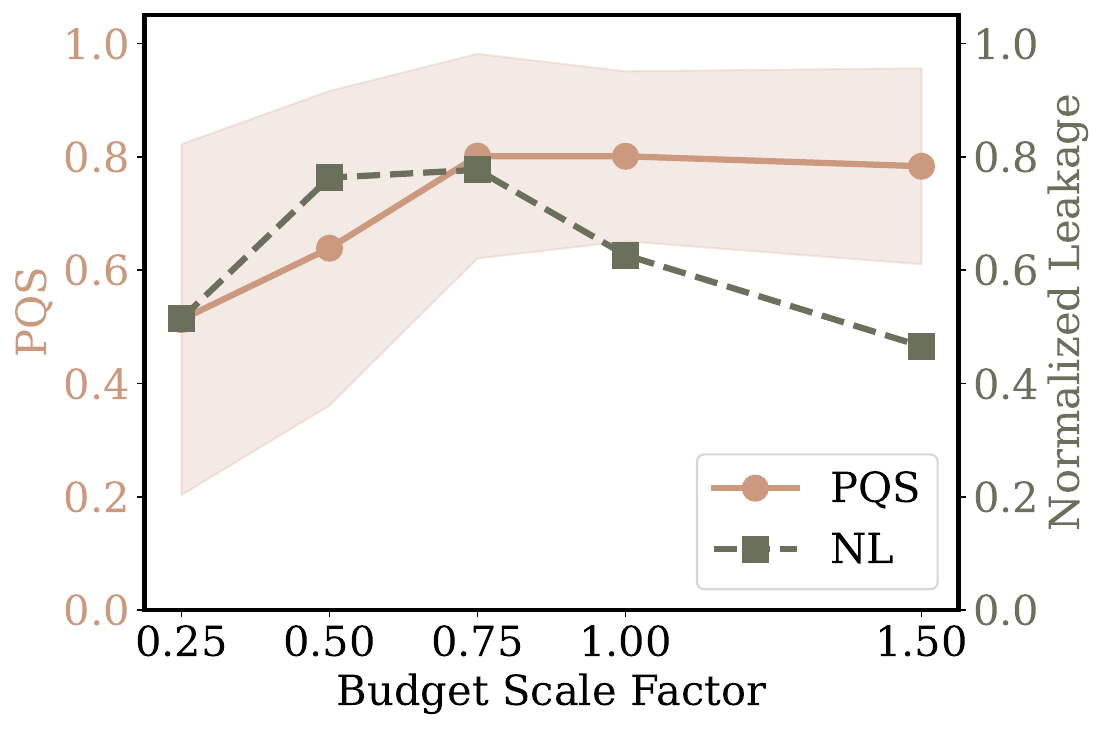}
\caption{Budget threshold sensitivity versus scale factor $\alpha$. PQS saturates at $\alpha \geq 0.75$, while normalized leakage peaks at $\alpha = 0.75$ and then declines as larger budgets leave more headroom. Shading denotes $\pm 1$ std.\ dev.\ in PQS.}
\label{fig:ablation-budget}
\end{figure}

\begin{center}
   \colorbox{theorycol!25}{
    \begin{minipage}{0.9\linewidth}
\textbf{Budget sensitivity.} PQS saturates at $\alpha \geq 0.75$, confirming that the default budget threshold is not over-restrictive. Reducing the budget to $\alpha = 0.25$ halves PQS, while $\alpha = 1.5$ provides no additional utility. The sweet spot $\alpha \in [0.75, 1.0]$ achieves PQS $>$ 0.80 with NL $<$ 0.78.
   \end{minipage}
} 
\end{center}

\subsection{Multi-Turn Budget Depletion}
\label{sec:eval:multiturn}

We evaluate the disclosure budget mechanism under multi-turn usage, where the middleware persists across all 58 tasks and budget consumption accumulates. This simulates a realistic deployment where an agent processes a sequence of tasks over a session, progressively depleting per-object budgets.

Fig.\ref{fig:ablation-multiturn-pqs} compares per-domain PQS. Overall PQS decreases from 0.773 (single-turn) to 0.767 (multi-turn), a modest $-$0.8\% degradation. The largest drops occur on coding ($-$0.233, from 0.753 to 0.520) and data pipeline ($-$0.167, from 1.000 to 0.833), where objects are accessed across multiple tasks, depleting their budgets and forcing the gatekeeper to withhold fields. Conversely, document (+0.133), security IR (+0.185), and database (+0.006) tasks show slight PQS \emph{increases} in multi-turn mode, because the planner benefits from residual context in the persistent twin graph.

Fig.\ref{fig:ablation-multiturn-nl} reveals the privacy side: multi-turn drives NL to 1.0 across all ten domains, compared to a mean of 0.611 in single-turn mode. This confirms that cumulative disclosure fully exhausts per-object budgets over a 58-task session. Importantly, SND remains 1.0 in both modes---the budget mechanism caps total leakage at $\tau(o)$ per object without ever exposing raw content. During the multi-turn run, 8 objects were excluded entirely (``budget exceeded at registration''), demonstrating the gatekeeper's enforcement of hard budget limits.

\begin{center}
   \colorbox{theorycol!25}{
    \begin{minipage}{0.9\linewidth}
\textbf{Multi-turn depletion} Multi-turn budget accumulation lowers PQS by 0.8\%, while NL reaches 1.0. The system degrades gracefully: over-budget objects are excluded rather than leaking additional information, and SND stays at 1.0 throughout.
   \end{minipage}
} 
\end{center}

\section{Discussion}
\label{sec:discussion}

We now discuss why \sysname provides privacy benefits beyond execution isolation alone, what leakage channels remain, how the design scales with environment and workflow size, and the main limitations of the current study.

\subsection{Privacy and Security Discussion}
\label{sec:discussion:security}

\sysname does not aim to make cloud planning \emph{cryptographically private}. Instead, it reduces exposure architecturally by constraining what the planner can observe through schema projection, bounded capabilities, and cumulative disclosure control.

\smallskip
\noindent\textbf{Content reconstruction.}
For raw-content reconstruction, the design is structurally restrictive. The cloud never receives raw files, free-form text, paths, credentials, or unsanitized execution outputs. All planner-visible state is produced by the projection pipeline and expressed through bounded schema fields.

As a result, the cloud-visible channel excludes raw content by construction rather than attempting to scrub it after the fact. Reconstruction therefore becomes possible only if the local projection itself emits overly specific attributes.

\smallskip
\noindent\textbf{Re-identification.}
For object re-identification (Goal~1), privacy depends not only on the schema vocabulary but also on how much of it is disclosed over time. As shown in \S\ref{sec:eval:reid}, disclosing all twin fields allows an adversary to recover object identity with high confidence. This means the twin is not intrinsically anonymous. Instead, privacy comes from \emph{partial observability}: the gatekeeper prevents the planner from accumulating a complete object fingerprint by enforcing per-object disclosure budgets. In this sense, the budget is a core security mechanism rather than a secondary optimization.

\smallskip
\noindent\textbf{Cross-session linkage.}
For cross-session linkage (Goal~2), fresh random object identifiers prevent trivial reuse of IDs across sessions. However, identifier randomization alone is insufficient. Repeated disclosure of rare semantic or relational patterns may still support linkage if the same object induces a distinctive signature across sessions. In practice, unlinkability is improved by the combination of session-local IDs, bounded vocabularies, and cumulative disclosure limits. The goal is not to make objects indistinguishable in principle, but to prevent the planner from observing enough stable attributes to link them reliably.

\smallskip
\noindent\textbf{Structural inference.}
For structural inference (Goal~3), \sysname reduces but does not eliminate leakage. Graph topology is itself informative: unusual dependency patterns, rare capability requirements, or asymmetric connectivity may reveal properties of the underlying project even when node attributes are sanitized. Bounded edge vocabularies and disclosure budgets reduce this signal, but cannot remove it entirely without also reducing planning utility. This is a structural trade-off rather than an implementation artifact.

\smallskip
\noindent\textbf{Prompt- and plan-level attack surfaces.}
The design also narrows two common attack surfaces. First, prompt-injection opportunities are reduced because the cloud planner never receives raw local artifacts, only schema-constrained abstractions. Second, unsafe or privacy-invasive plans are contained locally: every step is mediated by capability validation, policy checks, disclosure accounting, and optional human approval. Even if the planner proposes an overreaching plan, enforcement remains on the trusted local side.

\begin{figure}[!]
\centering
\includegraphics[width=0.8\columnwidth]{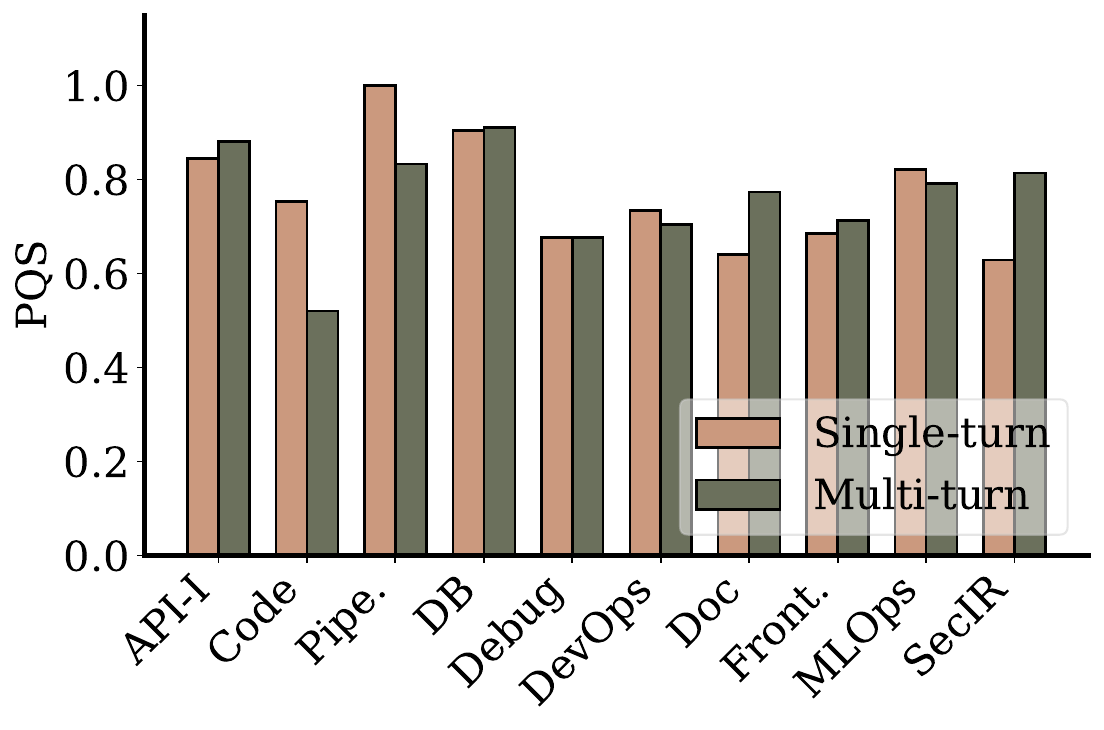}
\caption{Per-domain PQS under single- and multi-turn budgeting. Multi-turn most degrades data-pipeline and coding tasks. Repeated access to the same objects exhausts budgets.}
\label{fig:ablation-multiturn-pqs}
\end{figure}

\begin{figure}[!]
\centering
\includegraphics[width=0.8\columnwidth]{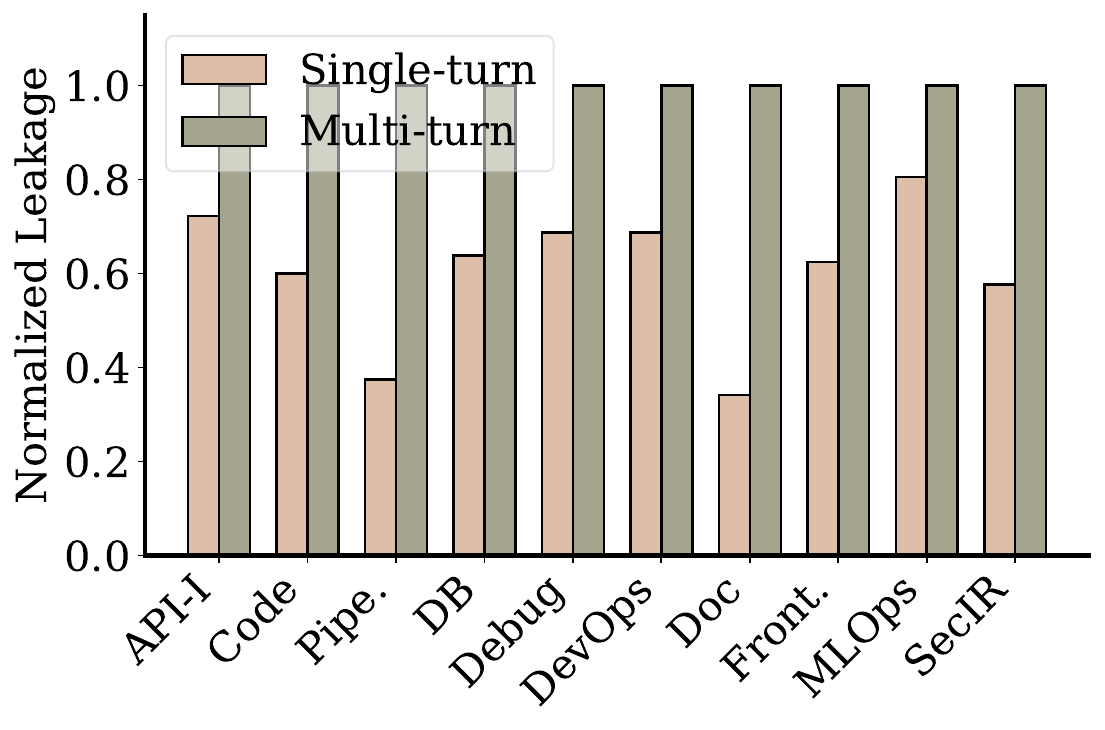}
\caption{Per-domain normalized leakage in single-/multi-turn modes. Multi-turn drives NL to 1.0 across domains.}
\label{fig:ablation-multiturn-nl}
\end{figure}

\subsection{Complexity and Scalability}
\label{sec:discussion:complexity}

Let $n$ be the number of tracked objects, $m$ the number of typed edges in the twin graph, and $k$ the number of plan steps. Let $|F_i|$ denote the token volume of object $i$, and let $d$ be the average number of retained schema fields per object.

\smallskip
\noindent\textbf{Projection cost.}
Stage~1 dominates local preprocessing when SLM-backed extraction is used, with cost
\[
O\!\left(\sum_i |F_i| \cdot c_{\text{SLM}}\right),
\]
where $c_{\text{SLM}}$ is the effective per-token extraction cost. Stages~2--4 are linear in the number of retained fields and relations,
\[
O(n d + m),
\]
since redaction, generalization, and schema projection operate over object-local metadata and typed graph structure.

\smallskip
\noindent\textbf{Budget and gatekeeping cost.}
Disclosure accounting scales with the number of newly referenced objects and fields per turn. Under bounded capability arity, this is effectively $O(1)$ per step and $O(k)$ per planning cycle; under a more general formulation it is upper bounded by $O(n)$ if every tracked object must be checked. Local validation and execution contribute $O(k\alpha)$, where $\alpha$ is the average cost of one validated capability execution.

\smallskip
\noindent \textbf{Cloud planning cost.}
Cloud planning is dominated by model inference over the sanitized twin and prompt state, which we write abstractly as $O(L)$, where $L$ captures planner inference cost. Because the twin is much smaller and more structured than raw context, this cost depends primarily on the size of the abstract state and the complexity of the requested workflow rather than the full volume of local artifacts.

Combining these terms, the per-cycle cost is
\[
O\!\left(\sum_i |F_i| \cdot c_{\text{SLM}} + nd + m + L + k\alpha\right).
\]

In practice, this yields two main regimes. For large local environments, Stage~1 extraction dominates because the local side must summarize many raw artifacts into schema form. For smaller environments or long-horizon workflows, cloud inference dominates. This matches the empirical results in \S\ref{sec:eval:overhead}, where total latency is driven primarily by SLM extraction and cloud API time, while deterministic projection and gatekeeping add negligible overhead.

\smallskip\noindent\textit{Scalability implication.}
The deterministic parts of the pipeline scale well, but end-to-end latency is bounded by whichever reasoning layer is active: local extraction for environment-heavy workloads, or cloud planning for reasoning-heavy workloads. Capability design and projection selectivity therefore matter not only for privacy, but also for performance.

\subsubsection{Toward Production Deployment.}
\label{sec:eval:deployment}

The current prototype transmits the sanitized twin $Z_t$ as a single JSON payload to the cloud planner, which reasons declaratively over the abstract graph without any execution environment.
A natural next step is \emph{cloud-side sandbox construction}: upon receiving $Z_t$, the cloud instantiates an ephemeral container whose directory structure, file stubs, and type annotations mirror the schema graph.
Concretely, each node in the typed abstract graph $G_t$ maps to a synthetic artifact (e.g., a code file containing only the declared \texttt{semantic\_class}, import stubs matching the \texttt{contains} tags, and size-appropriate placeholder bodies).
The planner then iterates with this sandbox twin, running linters, test harnesses, and static analyzers before producing a verified plan.

This design raises a synthetic-population privacy trade-off: the sandbox must contain enough semantic structure for meaningful tool execution, yet each synthetic artifact is derived solely from the bounded schema fields already disclosed under the budget $\tau(o)$.
We give a concise privacy argument showing that the sandbox construction preserves all existing guarantees.

\begin{proposition}[Synthetic-population privacy bound]
\label{prop:sandbox-privacy}
Let $\operatorname{Synth}: \mathcal{Z} \to \mathcal{S}$ be a (possibly randomized) sandbox generator that takes as input only the sanitized twin $Z_t = \Pi(X_t)$.
Let $R_{1:T}$ denote the sequence of tool outputs produced by the planner interacting with the sandbox content $S = \operatorname{Synth}(Z_t)$, where the planner's interaction policy is a function of $S$ alone and does not incorporate auxiliary knowledge correlated with $X_t$ beyond what $Z_t$ already discloses.
Then:
\begin{equation}
X_t \to Z_t \to S \to R_{1:T}.
\end{equation}
By the Markov property (conditional independence):
\begin{equation}
I(X_t;\, S, R_{1:T} \mid Z_t) = 0,
\end{equation}
i.e., the sandbox and all interaction transcripts are conditionally independent of the real environment given $Z_t$.
Consequently:
\begin{enumerate}
\item The $(k, \delta)$-anonymity bound (Theorem~\ref{thm:anonymity}) holds with identical parameters, since the adversary's posterior over object identity is fully determined by $Z_t$.
\item The $\epsilon$-unlinkability bound (Theorem~\ref{thm:unlinkability}) holds by the post-processing property of $D_\infty$: for any (randomized) function $f$, $D_\infty(f(M(o)) \| f(M(o'))) \leq D_\infty(M(o) \| M(o'))$ (see Appendix~\ref{app:proofs}, Theorem~\ref{thm:unlinkability}).
\item The composition bound (Theorem~\ref{thm:composition}) is unchanged, since sandbox interactions consume no additional budget beyond the fields already disclosed in $Z_t$.
\end{enumerate}
\end{proposition}

Two assumptions matter here. First, $\operatorname{Synth}$ uses only $Z_t$: it has no access to the real environment, no external knowledge base conditioned on the user's identity, and no retrieval-augmented generation over corpora that could correlate with the original content.
Second, the planner's interaction policy within the sandbox must not incorporate side information correlated with $X_t$ beyond what $Z_t$ discloses (e.g., through pretraining memorization of specific codebases).
Violations of either assumption break the Markov property and require separate analysis.

From an engineering perspective, lightweight sandbox runtimes such as E2B~\cite{e2b2025sandbox} (cold start ${\sim}$150\,ms, per-request ephemeral lifecycle) or gVisor-based isolation~\cite{skypilot2025sandbox} provide the necessary infrastructure.
The sandbox is destroyed after each planning cycle, ensuring no persistent state accumulates across sessions.
Integrating sandbox construction with the existing MCP capability interface~\cite{anthropic2024mcp} follows naturally from the architecture: each declared capability maps to a sandbox-side tool endpoint, and the gatekeeper validates sandbox tool invocations against the same policy as direct execution.
Stateful tools and asynchronous invocations require additional scoping logic that we do not address here.
We leave empirical evaluation of sandbox-assisted planning quality and the design of $\operatorname{Synth}$ generators that maximize tool utility within the privacy bound of Proposition~\ref{prop:sandbox-privacy} to future work.

\section{Conclusion}
\label{sec:conclusion}

We presented \sysname, a system for privacy-preserving cloud-assisted LLM planning. Instead of exposing raw local context to a strong remote planner, \sysname restricts planning to a bounded abstract view generated locally and enforced through capability mediation and multi-turn disclosure budgets.
\sysname achieves full sensitive-item non-disclosure while maintaining planning quality close to full-context cloud planning across 60 tasks, outperforming both simple redaction and local weak-model baselines in privacy-utility trade-offs. Our design is planner-agnostic and remains robust under ablations and budget-constrained multi-turn interaction.

Our results show that strong cloud planning does \textbf{NOT} require raw-context exposure. Planning-time observability should therefore be treated as a first-class systems interface in privacy-aware agents.

\begin{table}[t]
\caption{Notation.}
\label{tab:notation}
\centering
\vspace{3pt}
\begin{threeparttable}
\resizebox{\columnwidth}{!}{%
\begin{tabular}{c|l}
\toprule
\textbf{Symbol} & \multicolumn{1}{c}{\textbf{Meaning}} \\
\midrule
\rowcolor{kindsemcol!10}
  $X_t$ & Full local context at time $t$ \\
  $Z_t$ & Sanitized digital twin (= $G_t$) \\
\rowcolor{kindsemcol!10}
  $\Pi(\cdot)$ & Privacy projection operator (4-stage pipeline) \\
  $\piout$ & Output sanitizer (Stages 2--4 on execution results) \\
\rowcolor{kindsemcol!10}
  $G_t = (\mathcal{V}_t, \mathcal{E}_t, \mathcal{A}_t)$ & Typed abstract graph (nodes, edges, attributes) \\
  $\mathcal{C} = \{c_1, \ldots, c_n\}$ & Local capability catalog \\
\rowcolor{kindsemcol!10}
  $P_t$ & Cloud-generated plan at time $t$ \\
  $\Phi(\cdot)$ & Cloud planner function \\
\rowcolor{kindsemcol!10}
  $Y_t$ & Sanitized execution result \\
  $\Psi(\cdot)$ & Local gatekeeper + executor \\
\rowcolor{kindsemcol!10}
  $\Omega$ & Policy set \\
$ B(o)$ & Cumulative disclosure budget for object $o$ \\
\rowcolor{kindsemcol!10}
 $\Delta_t(o)$ & New disclosure at turn $t$ for object $o$ \\
 $\tau(o)$ & Disclosure threshold for object $o$ \\
\rowcolor{kindsemcol!10}
 $\mathcal{L}$ & System-wide leakage: $\max_o B(o)$ \\
 $\epsilon$ & Privacy budget bound \\
\rowcolor{kindsemcol!10}
 $w_f$ & Base disclosure cost of schema field $f$ \\
 $\text{grain}(\mathcal{C})$ & Avg.\ weighted schema fields per capability call \\
\bottomrule
\end{tabular}%
}
\begin{tablenotes}
\small
\item[-] \textit{Upper part}: system components. \textit{Lower part}: privacy analysis.
\end{tablenotes}
\end{threeparttable}
\end{table}

\bibliographystyle{ACM-Reference-Format}
\bibliography{refs}

\appendix

\section{Disclosure Cost Accounting}
\label{app:disclosure}

The weighted field-cost model assigns each schema field a base cost $w_f$ reflecting its re-identification potential. For an object $o$ with attribute set $A(o)$, the initial disclosure cost is
\begin{equation}
\Delta_0(o) = \sum_{f \in A(o)} w_f.
\end{equation}

For subsequent turns, only new reveals incur additional cost:
\begin{equation}
\Delta_t(o) = \sum_{f \in A_t(o) \setminus A_{1:t-1}(o)} w_f,
\end{equation}
where $A_t(o)$ is the set of fields disclosed at turn $t$ and $A_{1:t-1}(o)$ is the set disclosed previously. Thus, retransmitting the same twin state without new fields incurs zero additional disclosure cost.

\section{Security Proofs}
\label{app:proofs}

This appendix formalizes the security properties stated as design goals in \S\ref{sec:system:threat}. We prove $(k,\delta)$-anonymity under schema projection (Theorem~\ref{thm:anonymity}), $\epsilon$-unlinkability under fresh session IDs (Theorem~\ref{thm:unlinkability}), and a composition bound on cumulative disclosure across $T$ turns (Theorem~\ref{thm:composition}).

\subsection{Notation and Setup}

For the proofs below, we abstract away implementation details and work directly with the projected attribute space induced by the planner-visible abstraction.

Let $\mathcal{O} = \{o_1, \ldots, o_n\}$ be the set of real-world objects in the local environment. The schema projection $\Pi$ maps each object to a planner-visible representation $z = \Pi(o)$ drawn from a finite attribute space
\[
\mathcal{Z} = \mathcal{K} \times \mathcal{S} \times 2^{\mathcal{T}} \times 2^{\mathcal{U}} \times \mathcal{F} \times \mathcal{R},
\]
where $\mathcal{K}$ is the set of kind labels ($|\mathcal{K}| = 6$), $\mathcal{S}$ is the set of semantic-class labels, $\mathcal{T}$ is the set of content tags, $\mathcal{U}$ is the set of usability tags, $\mathcal{F}$ is the set of freshness buckets ($|\mathcal{F}| = 3$), and $\mathcal{R}$ is the set of sensitivity levels ($|\mathcal{R}| = 3$).

Let $\text{fields}(z)$ denote the set of fields present in $z$. For a representation $z$ with disclosed field subset $D \subseteq \text{fields}(z)$, define the \emph{equivalence class}
\begin{equation}
\text{EC}(z, D) = \{o' \in \mathcal{O} \mid \Pi(o')|_D = z|_D\},
\end{equation}
that is, the set of objects whose projected representation agrees with $z$ on all disclosed fields.

The disclosure budget for object $o$ after $T$ turns is $B_T(o) = \sum_{t=1}^{T} \Delta_t(o)$, capped by threshold $\tau(o)$ (Eq.~\ref{eq:budget}).

\paragraph{Auxiliary-knowledge model.}
For Theorem~\ref{thm:anonymity}, we model adversarial side information by a random variable $\mathcal{I}$. We assume a bounded-likelihood-ratio condition inside each equivalence class: there exists a bound $\eta \geq 1$ such that for every disclosed pattern $x$, every $o, o' \in \mathcal{O}_x := \text{EC}(x, D)$, and every realization $i$ of $\mathcal{I}$,
\begin{equation}
\label{eq:eta-bounded-sideinfo}
\frac{\Pr[\mathcal{I} = i \mid O = o]}{\Pr[\mathcal{I} = i \mid O = o']} \leq \eta.
\end{equation}
Thus, auxiliary knowledge may bias the posterior within an equivalence class, but cannot arbitrarily single out one member.

\paragraph{Cost calibration model.}
For Theorems~\ref{thm:unlinkability} and~\ref{thm:composition}, let $F_t$ denote the set of newly disclosed fields about object $o$ at turn $t$, so that $\Delta_t(o) = \sum_{f \in F_t} w_f$. For each field $f$, define its \emph{worst-case branching factor}
\begin{equation}
\label{eq:branching-factor}
b_f := \max_{D, x} \bigl|\{v \in \mathcal{V}_f : \exists\, o \in \text{EC}(x, D) \text{ with } \Pi_f(o) = v\}\bigr|,
\end{equation}
where $\mathcal{V}_f$ is the value domain of field $f$, and the maximum ranges over all previously disclosed field sets $D$ and patterns $x$ on $D$. We assume field costs are calibrated as
\begin{equation}
\label{eq:cost-calibration}
w_f \geq \ln b_f.
\end{equation}
Hence, revealing field $f$ can refine an existing equivalence class by at most a multiplicative factor $e^{w_f}$ in the worst case.

\subsection{Theorem 1: $(k,\delta)$-Anonymity under Schema Projection}

\renewcommand{\thetheorem}{\arabic{theorem}}
\setcounter{theorem}{0}
\begin{theorem}[Re-identification Bound]
\label{thm:anonymity}
Let the adversary observe $z^* = \Pi(o^*)$ with disclosed field set $D^*$ such that $\sum_{f \in D^*} w_f \leq \tau(o^*)$ (enforced by the gatekeeper's per-object budget). Assume:
\begin{enumerate}
\item the prior $\mu$ is uniform within each equivalence class, that is, conditional on $\Pi(O)|_{D^*} = x$, every object in $\mathcal{O}_x$ is equally likely a priori;
\item the adversary's side information $\mathcal{I}$ satisfies the bounded-likelihood-ratio condition~\eqref{eq:eta-bounded-sideinfo}.
\end{enumerate}
If $|\text{EC}(z^*, D^*)| \geq k$, then for any adversary $\mathcal{A}$,
\begin{equation}
\Pr[\mathcal{A}(z^*|_{D^*}, \mathcal{I}) = o^*] \leq \frac{\eta}{\eta + k - 1}.
\end{equation}
Equivalently,
\begin{equation}
\Pr[\mathcal{A}(z^*|_{D^*}, \mathcal{I}) = o^*] \leq \frac{1}{k} + \delta,
\end{equation}
where
\begin{equation}
\delta := \frac{\eta}{\eta + k - 1} - \frac{1}{k}
= \frac{(\eta - 1)(k - 1)}{k(\eta + k - 1)}.
\end{equation}
\end{theorem}

\begin{proof}
Fix the disclosed pattern $x := z^*|_{D^*}$ and write $E := \mathcal{O}_x = \text{EC}(z^*, D^*)$, $m := |E|$. By assumption, $m \geq k$.

\textit{Step 1 (Posterior inside the equivalence class).}
Since all objects in $E$ induce the same disclosed pattern $x$, Bayes' rule gives, for any $o \in E$ and side-information realization $i$,
\begin{equation}
\begin{split}
\Pr[O = o \mid \Pi(O)|_{D^*} = x, \mathcal{I} = i] = {}\\
\frac{
\Pr[\mathcal{I} = i \mid O = o]\,
\Pr[O = o \mid \Pi(O)|_{D^*} = x]
}{
\sum_{o' \in E}
\Pr[\mathcal{I} = i \mid O = o']\,
\Pr[O = o' \mid \Pi(O)|_{D^*} = x]
}.
\end{split}
\end{equation}
By the uniform-within-class prior assumption,
\[
\Pr[O = o \mid \Pi(O)|_{D^*} = x] = \frac{1}{m}
\]
for all $o \in E$. The $1/m$ factors cancel, yielding
\begin{equation}
\label{eq:posterior-sideinfo}
\Pr[O = o \mid \Pi(O)|_{D^*} {=} x, \mathcal{I} {=} i]
= \frac{\Pr[\mathcal{I} = i \mid O = o]}{\sum_{o' \in E} \Pr[\mathcal{I} = i \mid O = o']}.
\end{equation}

\textit{Step 2 (Bounding the maximum posterior mass).}
Let $\hat{o} \in E$ maximize the posterior in~\eqref{eq:posterior-sideinfo}. By the likelihood-ratio bound~\eqref{eq:eta-bounded-sideinfo}, for every $o' \in E \setminus \{\hat{o}\}$,
\[
\Pr[\mathcal{I} = i \mid O = o'] \geq \frac{1}{\eta}\Pr[\mathcal{I} = i \mid O = \hat{o}].
\]
Therefore,
\begin{equation}
\sum_{o' \in E} \Pr[\mathcal{I} = i \mid O = o']
\geq \Pr[\mathcal{I} = i \mid O = \hat{o}] + (m-1)\cdot \frac{1}{\eta}\Pr[\mathcal{I} = i \mid O = \hat{o}].
\end{equation}
Substituting into~\eqref{eq:posterior-sideinfo},
\begin{equation}
\max_{o \in E} \Pr[O = o \mid x, i]
\leq \frac{1}{1 + (m-1)/\eta}
= \frac{\eta}{\eta + m - 1}.
\end{equation}
Since $m \geq k$ and $u \mapsto \eta/(\eta + u - 1)$ is decreasing in $u$,
\begin{equation}
\max_{o \in E} \Pr[O = o \mid x, i] \leq \frac{\eta}{\eta + k - 1}.
\end{equation}

\textit{Step 3 (Optimal adversary success probability).}
Given $(x,i)$, the optimal re-identification strategy is MAP decoding, whose success probability equals the maximum posterior mass. Hence the bound applies to every adversary $\mathcal{A}$:
\[
\Pr[\mathcal{A}(z^*|_{D^*}, \mathcal{I}) = o^*] \leq \frac{\eta}{\eta + k - 1}.
\]
The equivalent form with $\delta$ follows immediately.
\end{proof}

\noindent\textit{Interpretation.}
When $\eta = 1$, the bound reduces to the classical $1/k$ anonymity guarantee. Empirically, the re-identification experiment in \S\ref{sec:eval:reid} corresponds to the unrestricted case in which all twin fields are disclosed; the budgeted setting limits the disclosed field set and enlarges the resulting equivalence class.

\subsection{Theorem 2: $\epsilon$-Unlinkability across Sessions}

\begin{theorem}[Cross-Session Unlinkability]
\label{thm:unlinkability}
Let sessions $s_1, s_2$ assign fresh random object IDs independently. For a target object $o$, let $D^{(1)}$ and $D^{(2)}$ be the disclosed field sets in the two sessions, and let $U := D^{(1)} \cap D^{(2)}$ be the overlap of fields visible in both sessions. Denote the overlap-pattern random variable $X_U := \Pi(O)|_U$, where $O$ is drawn from prior $\mu$. Assume $X_U$ is uniform on its support $\mathcal{Z}_U \subseteq \prod_{f \in U} \mathcal{V}_f$, that is,
\[
\Pr[X_U = x] = \frac{1}{|\mathcal{Z}_U|}
\]
for all $x \in \mathcal{Z}_U$. Then for any adversary $\mathcal{A}$,
\begin{equation}
\ln \frac{\Pr[\mathcal{A}(Z^{(1)}, Z^{(2)}) = 1 \mid \text{same}]}{\Pr[\mathcal{A}(Z^{(1)}, Z^{(2)}) = 1 \mid \text{diff}]} \leq \epsilon,
\end{equation}
where $\epsilon = \ln |\mathcal{Z}_U|$. Moreover, under the cost calibration assumption~\eqref{eq:cost-calibration} and budget enforcement $\sum_{f \in D^{(i)}} w_f \leq \tau(o)$,
\begin{equation}
\label{eq:eps-budget}
\epsilon \leq \sum_{f \in U} \ln b_f \leq \sum_{f \in U} w_f \leq \tau(o).
\end{equation}
\end{theorem}

\begin{proof}
Let $T_1, T_2$ denote the transcripts seen by the adversary in the two sessions, restricted to the overlap fields $U$. Fresh random IDs carry no cross-session information, so any linking rule depends only on the overlap attributes.

\textit{Step 1 (Distribution under ``same'' and ``diff'').}
The projection $\Pi$ is deterministic, so under the hypothesis \textsf{same} (both sessions involve the same object $O$), both transcripts produce the identical overlap pattern: $(T_1, T_2) = (X_U, X_U)$. Under the hypothesis \textsf{diff} (two independent objects $O_1, O_2 \sim \mu$), the overlap patterns are i.i.d.: $(T_1, T_2) = (X_U^{(1)}, X_U^{(2)})$ with $X_U^{(1)} \perp X_U^{(2)}$.

\textit{Step 2 (Max-divergence computation).}
For any adversary $\mathcal{A}$ with acceptance region $R \subseteq \mathcal{Z}_U \times \mathcal{Z}_U$,
\begin{equation}
\Pr[\mathcal{A} = 1 \mid \text{same}] = \sum_{x:\, (x,x) \in R} \Pr[X_U = x],
\end{equation}
\begin{equation}
\Pr[\mathcal{A} = 1 \mid \text{diff}] = \sum_{(x,y) \in R} \Pr[X_U = x]\Pr[X_U = y]
\geq \sum_{x:\, (x,x) \in R} \Pr[X_U = x]^2.
\end{equation}
The inequality holds because the off-diagonal terms in $R$ only increase the denominator. Under the uniform-on-support assumption, $\Pr[X_U = x] = 1/N$ where $N = |\mathcal{Z}_U|$. Hence
\begin{equation}
\frac{\Pr[\mathcal{A} = 1 \mid \text{same}]}{\Pr[\mathcal{A} = 1 \mid \text{diff}]}
\leq
\frac{\sum_{x:\, (x,x) \in R} 1/N}{\sum_{x:\, (x,x) \in R} 1/N^2}
= N = |\mathcal{Z}_U|.
\end{equation}
This bound is tight: the adversary $\mathcal{A}^*(x,y) = \mathbf{1}[x = y = x_0]$ for any fixed $x_0 \in \mathcal{Z}_U$ achieves ratio exactly $N$.

\textit{Step 3 (Budget connection).}
The support size satisfies $|\mathcal{Z}_U| \leq \prod_{f \in U} |\mathcal{V}_f|$. Under the cost calibration~\eqref{eq:cost-calibration}, $|\mathcal{V}_f| \leq b_f \leq e^{w_f}$, so
\begin{equation}
\epsilon = \ln |\mathcal{Z}_U|
\leq \sum_{f \in U} \ln |\mathcal{V}_f|
\leq \sum_{f \in U} w_f.
\end{equation}
Since $U \subseteq D^{(i)}$ for each session $i$, and the budget enforces $\sum_{f \in D^{(i)}} w_f \leq \tau(o)$, we obtain $\epsilon \leq \tau(o)$.
\end{proof}

\noindent\textit{Interpretation.}
The parameter $\epsilon$ is directly bounded by the overlap of disclosed fields across sessions, and therefore by the disclosure budget $\tau(o)$. Tighter budgets reduce the attribute overlap available for linkage. When no fields overlap across sessions, $\epsilon = 0$ and linkage through planner-visible attributes becomes impossible.

\noindent\textit{Remark (non-uniform distributions).}
For non-uniform $X_U$, the worst-case adversary achieves ratio $1/\min_{x \in \text{supp}} \Pr[X_U = x]$, giving $\epsilon = -\ln \min_x \Pr[X_U = x]$. This can exceed $\ln |\mathcal{Z}_U|$ when the distribution is skewed. The budget connection still bounds the number of distinguishable overlap patterns, even when their frequencies are not uniform.

\subsection{Theorem 3: Composition Bound for Multi-Turn Disclosure}

\begin{theorem}[Sequential Composition]
\label{thm:composition}
Consider an object $o$ observed across $T$ turns. Let $D_t$ be the cumulative disclosed field set after turn $t$, and let $\text{EC}_t := \text{EC}(\Pi(o), D_t)$ be the corresponding equivalence class. Assume the field costs satisfy the calibration condition~\eqref{eq:cost-calibration}. Then
\begin{equation}
\label{eq:composition-exp}
|\text{EC}_T| \geq |\text{EC}_0| \exp\!\Bigl(-\sum_{t=1}^{T} \Delta_t(o)\Bigr).
\end{equation}
In particular, if the gatekeeper enforces $B_T(o) = \sum_{t=1}^{T} \Delta_t(o) \leq \tau(o)$, then
\begin{equation}
\label{eq:composition-budget}
|\text{EC}_T| \geq |\text{EC}_0|\, e^{-\tau(o)}.
\end{equation}
\end{theorem}

\begin{proof}
For each turn $t$, let $F_t := D_t \setminus D_{t-1}$ be the set of newly disclosed fields, so that $\Delta_t(o) = \sum_{f \in F_t} w_f$.

\textit{Step 1 (Monotone shrinkage).}
Because $D_{t-1} \subseteq D_t$, every object consistent with the disclosures at turn $t$ is also consistent at turn $t-1$. Hence $\text{EC}_t \subseteq \text{EC}_{t-1}$ and $|\text{EC}_t| \leq |\text{EC}_{t-1}|$.

\textit{Step 2 (Per-field refinement bound).}
Fix turn $t$ and suppose one additional field $f \in F_t$ is revealed. Within the pre-existing class $\text{EC}_{t-1}$, disclosure of $f$ partitions that class into at most $b_f$ subclasses by definition of the branching factor~\eqref{eq:branching-factor}. The subclass selected by the true value of field $f$ therefore has size at least $|\text{EC}_{t-1}|/b_f$. Using the cost calibration~\eqref{eq:cost-calibration}, $b_f \leq e^{w_f}$, so revealing field $f$ shrinks the class by at most a factor $e^{w_f}$:
\begin{equation}
|\text{EC}^{(f)}| \geq |\text{EC}_{t-1}|\, e^{-w_f}.
\end{equation}
Applying this sequentially to all newly revealed fields in $F_t$ gives
\begin{equation}
\label{eq:per-turn-shrink}
|\text{EC}_t| \geq |\text{EC}_{t-1}| \exp\!\Bigl(-\sum_{f \in F_t} w_f\Bigr)
= |\text{EC}_{t-1}|\, e^{-\Delta_t(o)}.
\end{equation}

\textit{Step 3 (Sequential composition over $T$ turns).}
Iterating~\eqref{eq:per-turn-shrink} over $t = 1, \ldots, T$ yields
\begin{equation}
|\text{EC}_T| \geq |\text{EC}_0| \prod_{t=1}^{T} e^{-\Delta_t(o)}
= |\text{EC}_0| \exp\!\Bigl(-\sum_{t=1}^{T} \Delta_t(o)\Bigr),
\end{equation}
which proves~\eqref{eq:composition-exp}. Budget enforcement $B_T(o) \leq \tau(o)$ immediately gives~\eqref{eq:composition-budget}.
\end{proof}

\noindent\textit{Interpretation.}
Once the budget is exhausted, the gatekeeper either withholds the object from the planner-visible abstraction or refuses further field disclosures. This yields graceful degradation: the planner may lose utility, but the equivalence class is prevented from shrinking beyond the budget-implied bound.

\noindent\textit{Empirical alignment.}
The multi-turn experiment (\S\ref{sec:eval:multiturn}) is consistent with this bound: after 58 tasks, NL reaches 1.0 while several objects are excluded by the gatekeeper, demonstrating hard enforcement of the budget cap. PQS degrades only modestly, indicating graceful degradation rather than catastrophic failure.

\noindent\textit{Corollary (product-form bound).}
Since $e^{-x} \geq 1-x$ for $x \geq 0$, we also obtain
\begin{equation}
|\text{EC}_T| \geq |\text{EC}_0| \prod_{t=1}^{T} \bigl(1-\Delta_t(o)\bigr),
\end{equation}
whenever each $\Delta_t(o)$ is normalized to lie in $[0,1]$. This product form is weaker but sometimes more intuitive.

\subsection{Assumptions and Empirical Alignment}

Table~\ref{tab:proof-empirical} summarizes how each formal result connects to the empirical evaluation and highlights the main modeling assumption behind the corresponding bound.

\begin{table}[h]
\centering
\caption{Mapping between formal security properties, assumptions, and empirical evidence.}
\label{tab:proof-empirical}
\resizebox{\columnwidth}{!}{%
\begin{tabular}{c|l|l|l}
\toprule
\multicolumn{1}{c}{\textbf{Thm.}} &
\multicolumn{1}{c}{\textbf{Property}} &
\multicolumn{1}{c}{\textbf{Key assumption}} &
\multicolumn{1}{c}{\textbf{Empirical evidence}} \\
\midrule
\ref{thm:anonymity} & $(k,\delta)$-anonymity & $\eta$-bounded side info & \S\ref{sec:eval:reid}: 94.1\% w/o budget \\
\ref{thm:unlinkability} & $\epsilon$-unlinkability & Uniform on support & $\epsilon \leq \tau(o)$ via budget \\
\ref{thm:composition} & Composition bound & Cost calibration $w_f \geq \ln b_f$ & \S\ref{sec:eval:multiturn}: NL${=}$1.0 \\
\bottomrule
\end{tabular}%
}
\end{table}

The formal guarantees rely on three modeling assumptions:
\begin{itemize}[nosep]
\item \textit{Bounded side information within equivalence classes} (Theorem~\ref{thm:anonymity}). The likelihood-ratio bound~\eqref{eq:eta-bounded-sideinfo} limits how strongly auxiliary knowledge can skew the posterior within an equivalence class. When $\eta = 1$, the classical $1/k$ guarantee is recovered exactly.

\item \textit{Uniform overlap-pattern distribution} (Theorem~\ref{thm:unlinkability}). Cross-session linkage reduces to distinguishing overlap attribute patterns once fresh IDs remove identifier-based matching. The uniform assumption yields the clean bound $\epsilon = \ln |\mathcal{Z}_U|$; skewed distributions can only worsen it.

\item \textit{Field-cost calibration} (Theorems~\ref{thm:unlinkability} and~\ref{thm:composition}). The condition $w_f \geq \ln b_f$ ensures that each field cost upper-bounds the logarithmic refinement induced by disclosing that field. This links the unlinkability and composition results to the single budget parameter $\tau(o)$.
\end{itemize}

\end{document}